\documentclass[5p,twocolumn,10pt,times]{elsarticle}
\usepackage{amsmath}
\usepackage{hyperref}

\newtheorem{theorem}{\bf Theorem}
\newenvironment{proof}{\paragraph{Proof:}}{\hfill$\square$}

\usepackage{booktabs}
\usepackage{multirow}
\usepackage{graphicx}

\usepackage{color}
\usepackage{wrapfig}
\usepackage[percent]{overpic} 
\definecolor{AAA}{rgb}{1.0, 0.13, 0.32}
\definecolor{BBB}{rgb}{0.2, 0.1, 1}
\definecolor{ccc}{rgb}{0, 0.5, 0.25}

\begin{document}
\baselineskip11pt

\begin{frontmatter}

\title{Toward Precise Curve Offsetting Constrained to Parametric Surfaces
}

\author[a]{Jin Zhao\fnref{equal1}}
\author[a]{Pengfei Wang\fnref{equal1}}
\author[b]{Shuangmin Chen}
\author[a]{Jiong Guo}
\author[a]{Shiqing Xin\corref{cor1}}
\author[a]{Changhe Tu}
\author[c]{Wenping Wang}

\fntext[equal1]{These authors contributed equally to this work.}
\cortext[cor1]{Corresponding author}

\address[a]{School of Computer Science and Technology, Shandong University, Qingdao, Shandong 266237, China}
\address[b]{College of Information Science and Technology, Qingdao University of Science and Technology, Qingdao, Shandong 266061, China}
\address[c]{Department of Computer Science and Engineering, Texas A\&M University, College Station, TX 77843, USA}

\begin{abstract} 
Computing offsets of curves on parametric surfaces is a fundamental yet challenging operation in computer-aided design and manufacturing. Traditional analytical approaches suffer from time-consuming geodesic distance queries and complex self-intersection handling, while discrete methods often struggle with precision. In this paper, we propose a totally different algorithm paradigm. Our key insight is that by representing the source curve as a sequence of line-segment primitives, the Voronoi decomposition constrained to the parametric surface enables localized offset computation. Specifically, the offsetting process can be efficiently traced by independently visiting the corresponding Voronoi cells. To address the challenge of computing the Voronoi decomposition on parametric surfaces, we introduce two key techniques. First, we employ intrinsic triangulation in the parameter space to accurately capture geodesic distances. Second, instead of directly computing the surface-constrained Voronoi decomposition, we decompose the triangulated parameter plane using a series of plane-cutting operations. Experimental results demonstrate that our algorithm achieves superior accuracy and runtime performance compared to existing methods. We also present several practical applications enabled by our approach.

\end{abstract}


\end{frontmatter}

\section{Introduction}

The offsetting operation of curves on surfaces (as illustrated in Figure~\ref{FIG:example}) is a fundamental yet challenging geometric operation in computer-aided design and manufacturing, where both accuracy and run-time performance are crucial. It has garnered significant attention due to its wide range of applications, including surface blending, surface coverage, and path planning~\cite{doi:10.1177/0278364905059058,10.1145/77055.77057,LEE2003511}. Even if the source curve is smooth and defined in a 2D plane, the resulting offset may include cusps. When the base domain is a curved surface, the complexity increases due to the intrinsic curvature of the surface geometry.

To be detailed, the challenges arise from the following aspects. First, surface-constrained offsets require computing geodesic distances, but geodesic distances are significantly more computationally expensive than straight-line distances. Second, the offsetting operation can be explained based on the Minkowski sum, which essentially ``expands'' the source curve by a fixed distance in every direction, creating a ``thickened'' version of the original curve. Therefore, there may not exist a one-to-one mapping between them, leading to cusps in the offset.

\begin{figure}[h]
	\centering
\begin{overpic}
[width=0.6\linewidth]{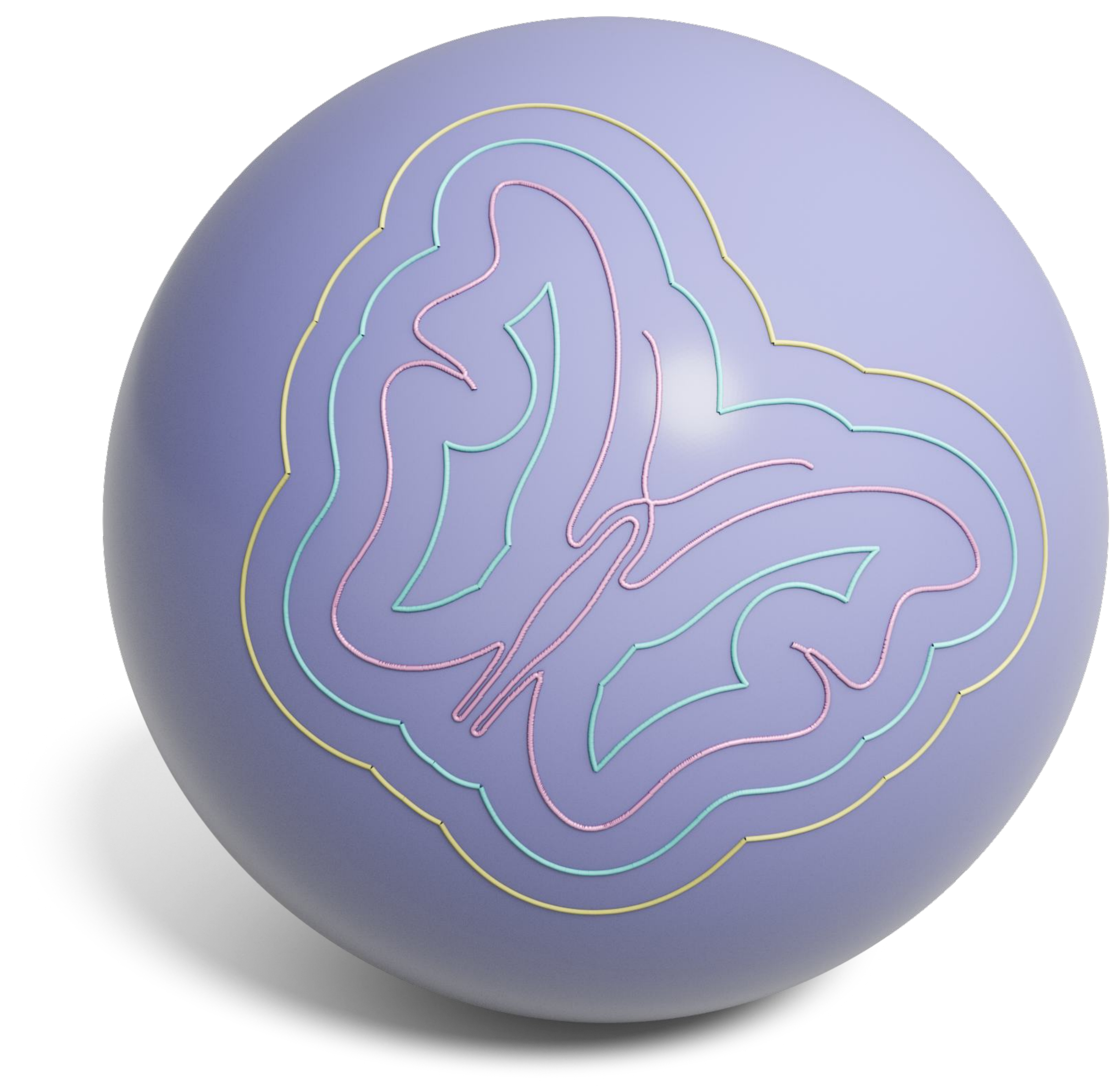}
\end{overpic}
\vspace{2mm}
\caption{Offset curves on a spherical parametric surface. The source curve (shown in pink) and its offset curves at two different distances (shown in cyan and yellow, respectively) are illustrated on the spherical surface.
}
\label{FIG:example}
\end{figure}

Current approaches typically employ a two-phase strategy: initially generating raw offset curves via geodesic direction displacement, followed by resolving self-intersections through geometric operations. Although conceptually elegant, this bifurcated approach introduces substantial practical challenges. The generation of raw offsets via numerical integration or iterative methods imposes considerable computational cost, while the resolution of self-intersections frequently exhibits numerical instability, especially in regions characterized by high curvature.

Moreover, despite geodesic computation being an intrinsic property of parametric surfaces—independent of their spatial embedding—existing methods rarely leverage this intrinsic characteristic. This deficiency presents an opportunity to develop more principled approaches that operate directly on the surface's intrinsic geometry, thereby enhancing both computational accuracy and algorithmic efficiency.

In this paper, we propose a novel approach to geodesic curve offsetting on parametric surfaces. Our method is built on two key insights: First, we recognize that geodesic distance computation is an intrinsic problem independent of the parametric surface's embedding. Consequently, we focus on the parametric space with an induced metric rather than the embedded parametric surface. We extend existing geodesic distance computation methods to calculate distances between arbitrary point pairs on the parametric space. Second, we observe that when calculating offsets and Voronoi diagrams of several primitives, the offset of a specific primitive is constrained to its Voronoi cell. Based on this observation, we discretize the source curve into multiple segments, calculate the Voronoi diagram of these segments, and then extract the offset of each segment within its corresponding cell. Through these two insights, our approach achieves efficient extraction and composition of complete offset curves while naturally handling self-intersections.

The key innovations of our method are:

\begin{itemize}
\item \textbf{Parameter Space Geodesic Computation}: We propose a novel approach that computes geodesics by operating directly in the parameter space equipped with a surface-induced metric.
\item \textbf{Voronoi-Guided Decomposition}: Through theoretical analysis of offset curves and Voronoi diagrams, we establish that offset regions remain confined within their corresponding Voronoi cells, enabling localized and parallel computation.
\item \textbf{No Specialized Self-Intersection Handling}: Our method does not require specialized handling for self-intersection issues. 
\end{itemize}

\section{Releated Work}
\subsection{Geodesics Distance Calculation}
Existing approaches for computing geodesics on parametric surfaces can be broadly classified into three categories: analytical methods, numerical approaches, and discrete approximation methods.

\subsubsection{Analytical Methods} 
The analytical approaches, as presented by Do Carmo~\cite{do2016differential}, offer theoretically elegant solutions but encounter substantial practical limitations. These methods exhibit considerable complexity, and closed-form solutions for geodesics cannot be derived for general surfaces. This fundamental constraint has necessitated the development of more pragmatic computational methodologies.

\subsubsection{Numerical Methods} 
Numerical methods have gained widespread adoption in recent decades. Beck et al.~\cite{Beck198618} computed geodesic paths on bicubic spline surfaces using fourth-order Runge–Kutta methods. Patrikalakis and Badris~\cite{Patrikalakis198939} examined geodesic curves on parametric surfaces during their construction of offset curves on Rational B-spline surfaces. Sneyd and Peskin~\cite{doi:10.1137/0911014} investigated geodesic path computation on generalized cylinders employing second-order Runge–Kutta methods. For shortest path problems, Maekawa~\cite{10.1115/1.2826919} introduced an approach based on relaxation methods utilizing finite difference discretization, while Kasap et al.~\cite{KASAP20051206} presented a methodology for solving nonlinear differential equations through finite-difference and iterative techniques.

\subsubsection{Discrete Approximation Methods} 
The third approach involves first converting parametric surfaces into mesh surfaces and then applying mesh-based geodesic computation methods. Many accurate discrete methods approach geodesics or the shortest paths on tessellated surfaces~\cite{tucker1997forming,RAVIKUMAR2003119}, polygonal surfaces~\cite{polthier2006straightest,KANAI2001801} and triangular meshes~\cite{MARTINEZ2005667,Surazhsky2005553}.
Recent work~\cite{10.1145/3414685.3417839} achieves efficient geodesic computation through edge flipping operations, demonstrating significant computational performance. 
As comprehensively reviewed by Bose et al.~\cite{BOSE2011486}, these algorithms are differentiated based on theoretical time complexity and approximation ratio. Discrete geodesics have been gaining attention as computers become increasingly more powerful and discretized models become more prevalent in geometric modeling. However, the discrete geodesics cannot be computed directly on the original smooth surface, which limits their application in scenarios requiring high accuracy, and the conversion from parametric surfaces to mesh surfaces introduces significant precision loss.

\vspace{5mm}
Despite the numerous existing approaches for geodesic computation, it remains challenging to efficiently compute geodesic distances between arbitrary points on parametric surfaces while maintaining high accuracy for applications requiring fast query responses.

\subsection{Geodesic Offset Curve on Surface}
The problem of computing offset curves has a long history of development in geometric modeling and processing~\cite{10.1145/116873.116880,586019,MAEKAWA1998437,PHAM1992223}, due to its essential role in applications such as toolpath generation for CNC machining~\cite{10.5555/108340} and geometric tolerancing. Since our work focuses on geodesic offset computation on surfaces, we mainly review previous approaches for computing geodesic offset curves.

The study of geodesic offset curves was pioneered by Patrikalakis and Bardis~\cite{Patrikalakis1989OffsetsOC}, who introduced this problem to geometric modeling. Wolter and Tuohy~\cite{10.1007/BF01200103} later revisited this problem as a special case of approximating procedurally defined curves on surfaces. Brunnet~\cite{10.5555/647587.730878} investigated geodesic offsets specifically on surfaces of revolution, reducing the problem to solving zeros of complex integral functions. Due to the computational complexity, these approaches typically employed the Runge-Kutta scheme for efficient approximation of geodesic offset curves.

Ulmet~\cite{GeodesicOffsetsOfSpline} proposed two simplified alternatives to geodesic offset computation. The first approach computes a planar offset curve in the parameter domain and maps it to the surface. However, this method suffers from parametrization-dependent issues - regions of the same size in the $(u,v)$-domain can map to patches of considerably different sizes on the surface, leading to non-uniform gaps between the curve and its offset. The second method projects a spatial offset curve onto the surface, which presents a non-trivial computational challenge. Tam et al.~\cite{Tam2004AGA} attempted to simplify this by intersecting the surface with the normal plane of the curve and then moving along the cross-sectional curve by the offset distance. While conceptually simpler than surface projection, this method proved unsuitable for generating continuous offset curves on surfaces.

Beyond the basic offset computation, a significant challenge lies in the handling of self-intersections. Most methods perform trimming in the $(u,v)$-parameter domain~\cite{Tam2004AGA} or in flattened planar domains for mesh surfaces~\cite{XU2015131,doi:10.1177/0954405413492965}. However, due to numerical errors in offset curve sampling and the computational difficulty in measuring geodesic distances on the surface, the offset trimming can be highly unstable, particularly when offset curves have tangential intersections. For applications requiring constant scallop-height maintenance~\cite{FENG2002647,10.1115/1.2826244,10.1115/1.2901938}, the offset distance is computed as a function $d(t)$, with different formulations depending on each method's approximation of local surface geometry.

\vspace{5mm}
Compared with the numerous research results in Euclidean spaces~\cite{10845125}, there are relatively few studies on geodesic offset computation for surfaces~\cite{10.5555/647587.730878,Patrikalakis1989OffsetsOC,ComputationofMedialCurvesinSurfaces,10.1007/BF01200103}, while most work on geodesic offset curves has focused on polygonal surfaces such as triangular meshes~\cite{HOLLA20031099,doi:10.1073/pnas.95.15.8431,7102776,XIN20111468}. These limitations in existing methods motivate the need for more efficient approaches that can handle geodesic offset computation directly on parametric surfaces while maintaining robust handling of self-intersections.

\subsection{Intrinsic Triangulation}
\label{sec:Intrinsic Triangulation}

Intrinsic triangulation fundamentally differs from its extrinsic counterpart by focusing solely on the intrinsic properties of the surface, independent of any particular embedding~(Figure.~\ref{FIG:intrinsicTriangulation}). It can be completely characterized by three key components: (1) the mesh connectivity describing how vertices are connected to form triangles, (2) the lengths of edges connecting adjacent vertices, and (3) the triangle inequality constraints that these distances must satisfy. These edge lengths and connectivity patterns fully define the triangulation without reference to vertex positions in any ambient space.

\begin{figure}[h]
	\centering
\begin{overpic}
[width=0.98\linewidth]{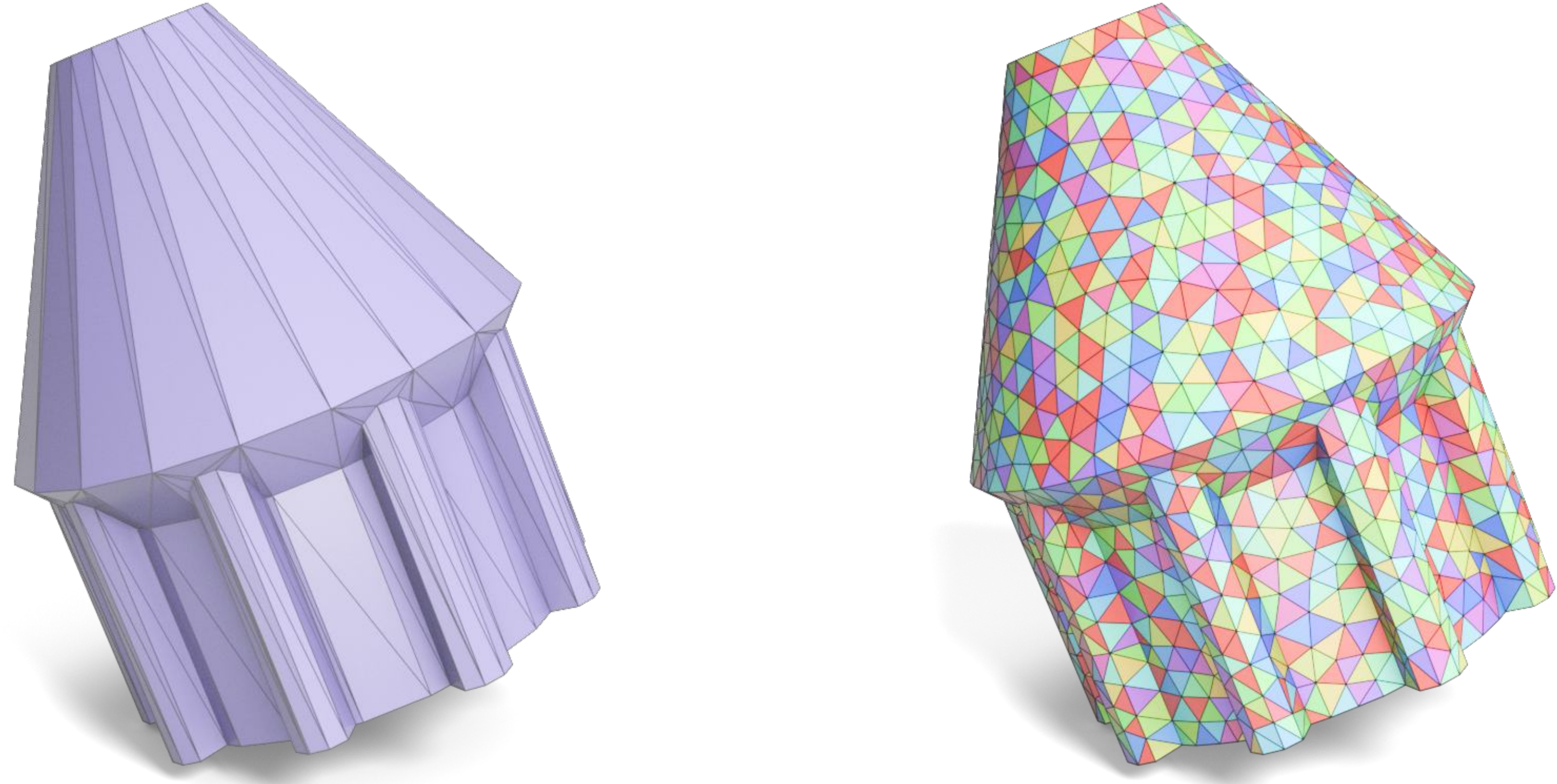}
\end{overpic}
\vspace{2mm}
\caption{The input mesh (left) and its intrinsic triangulation (right, with different intrinsic triangles shown in distinct colors). An intrinsic triangulation is fully defined by its mesh connectivity and edge lengths, which satisfy triangle inequalities, rather than depending on vertex positions in space.}
\label{FIG:intrinsicTriangulation}
\end{figure}

Sharp et al.~\cite{Sharp:2019:NIT} introduced a novel data structure that efficiently represents and manipulates such intrinsic triangulations. This approach is particularly valuable for computational tasks that rely only on intrinsic surface properties, such as geodesic distance computation and vector field processing. For meshes with poor quality elements that would typically hinder numerical computations, intrinsic triangulation provides a unique solution by improving mesh quality without altering the underlying geometry. This stands in stark contrast to traditional remeshing approaches, which inevitably must balance element quality against geometric fidelity.

Furthermore, intrinsic triangulation offers several key advantages in computational settings. For finite element methods, it simultaneously provides accurate geometry representation and high-quality elements for computation, since all geometric quantities can be derived purely from edge lengths. Its ability to maintain geometric fidelity while enabling more accurate discretization makes it especially suited for applications ranging from geodesic computations to vector field processing. Most notably, this approach eliminates the traditional trade-off between mesh quality and geometric accuracy, offering a powerful tool for processing geometrically complex or poorly tessellated surfaces.

\section{Methods}
\subsection{Geodesic Computation via Parameter Space Triangulation}
\label{sec:Geodesic Computation via Parameter Space Triangulation}

Computing geodesics on parametric surfaces is fundamentally an intrinsic problem—it depends solely on the surface's metric structure, independent of the surface's spatial embedding. However, most existing approaches fail to fully exploit this intrinsic characteristic. A prevalent methodology discretizes the parametric surface into a triangle mesh before applying mesh-based algorithms such as the edge flipping method of Sharp and Crane~\cite{10.1145/3414685.3417839}, as illustrated in Figure~\ref{FIG:flip} (a)(b). While this remeshing approach facilitates efficient computation, it introduces unnecessary discretization artifacts by compromising the surface's intrinsic metric structure.

\begin{figure}[h]
	\centering
\begin{overpic}
[width=0.97\linewidth]{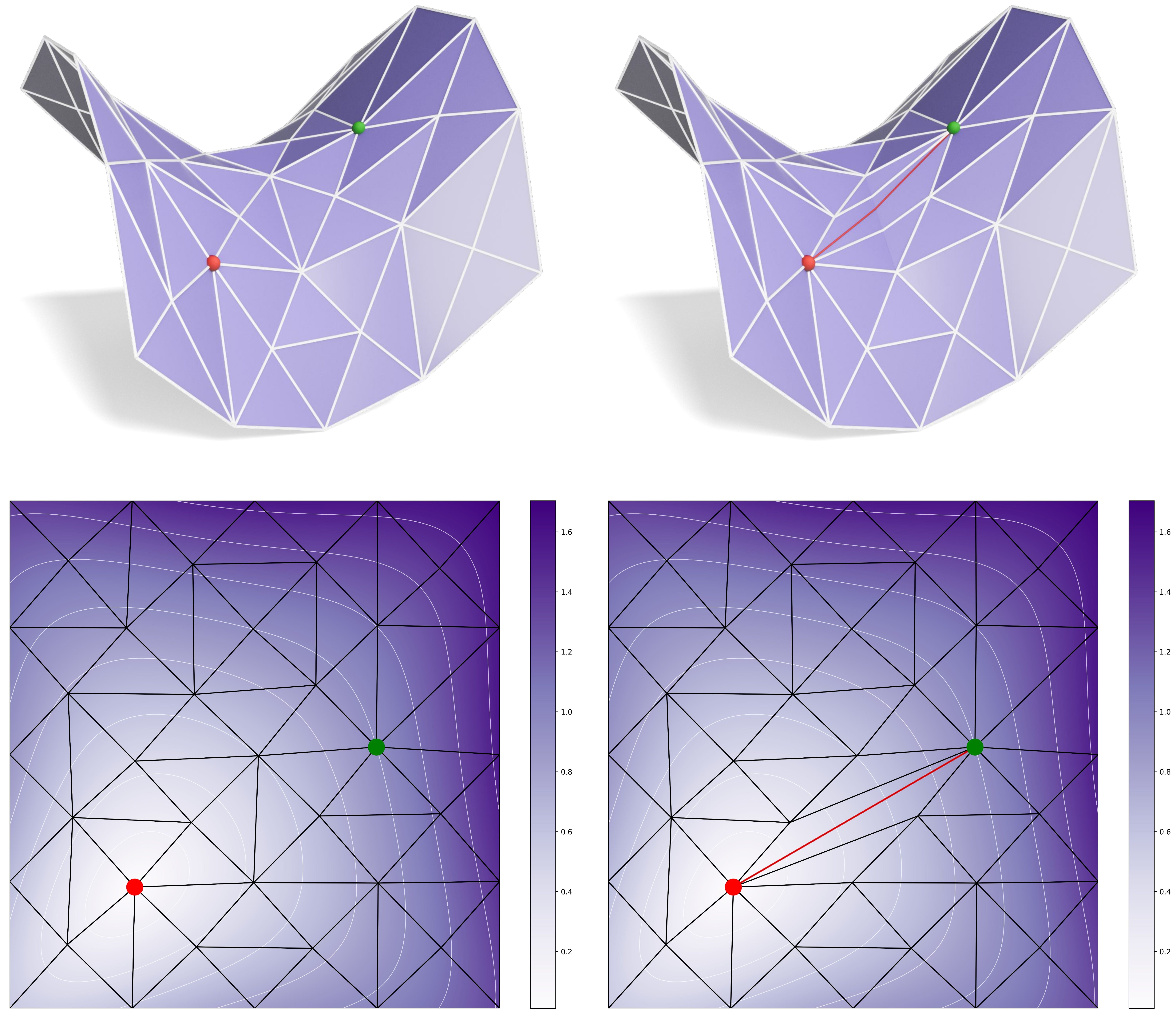}
\put(20,45){(a)}
\put(73,45){(b)}
\put(20,-5){(c)}
\put(73,-5){(d)}
\end{overpic}
\vspace{2mm}
\caption{Computing geodesic paths between a start point (red) and end point (green) on a parametric surface. A common approach is to first discretize the parametric surface into a triangle mesh (a) and then apply mesh-based geodesic computation methods such as~\cite{10.1145/3414685.3417839} that use edge flipping. However, this discretization process introduces additional errors that affect the accuracy of geodesic distance computations. Instead, we directly compute an intrinsic triangulation of the parametric surface in the parameter domain (c), where each edge is assigned a length corresponding to the geodesic distance between its endpoints on the surface. The geodesic distance can then be accurately computed through intrinsic triangle flipping operations (d) directly in the parameter space. The color gradient visualizes the geodesic distance field from the start point, with contour lines representing equidistant paths.
}
\label{FIG:flip}
\end{figure}
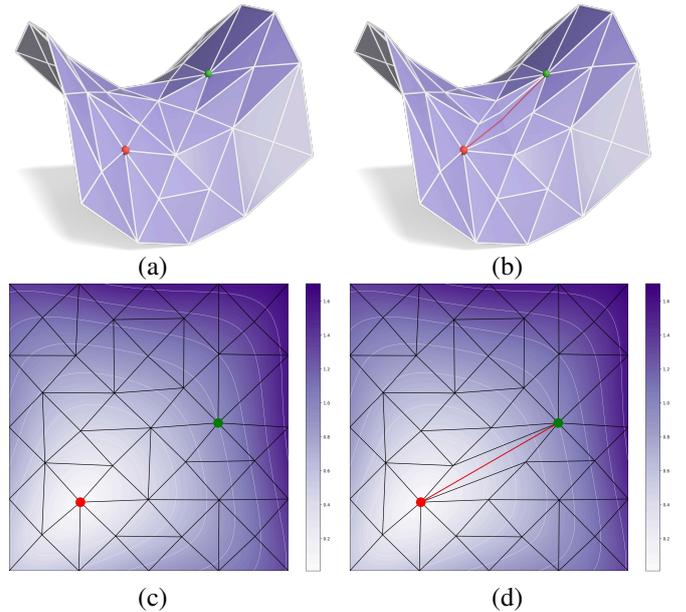
Rather than compromising accuracy through mesh approximation, we propose to operate directly in the parameter space while preserving the intrinsic geometric properties of the surface. For the parametric space, we introduce an induced metric—the distance between any two points in the parametric space is defined as the shortest paths distance between their corresponding points on the parametric surface. For computational convenience, we construct an intrinsic triangulation of the parametric space, where edge lengths are determined by the induced metric.

We then extend the intrinsic edge-flipping algorithm of Sharp and Crane~\cite{10.1145/3414685.3417839} to operate within our metric-equipped parameter space, as shown in Figure~\ref{FIG:flip}(c)(d). This extension allows the computation of geodesic distances between arbitrary points through intrinsic triangle flipping operations directly in the parameter domain.

The theoretical foundation of our method is predicated on the observation that this metric space naturally admits triangulations
where edges represent shortest paths. We establish the following fundamental property to validate our parameter space construction:

\begin{theorem}
For any three points $p_1$, $p_2$, and $p_3$ in the parameter space equipped with the induced metric, their distances satisfy the triangle inequality:
$$d_g(p_1,p_2) + d_g(p_2,p_3) \geq d_g(p_1,p_3)$$   
where $d_g(p_i,p_j)$ denotes the shortest path distance between the surface points corresponding to parameters $p_i$ and $p_j$.
\end{theorem}

\begin{proof}
Let $\gamma_{12}$ and $\gamma_{23}$ be the shortest paths connecting $p_1$ to $p_2$ and $p_2$ to $p_3$ respectively. Let $\gamma_{13}$ be the shortest path from $p_1$ to $p_3$, as shown in Figure~\ref{FIG:triangle}. Then:
\begin{itemize}
    \item The concatenated path $\gamma_{12} \cup \gamma_{23}$ forms a valid path from $p_1$ to $p_3$.
    \item By definition, $\gamma_{13}$ is the shortest path between $p_1$ and $p_3$.
    \item Therefore: $d_g(p_1,p_3) = \text{length}(\gamma_{13}) \leq \text{length}(\gamma_{12}) + \text{length}(\gamma_{23}) = d_g(p_1,p_2) + d_g(p_2,p_3)$
\end{itemize}
\end{proof}

\begin{figure}[h]
	\centering
\begin{overpic}
[width=0.6\linewidth]{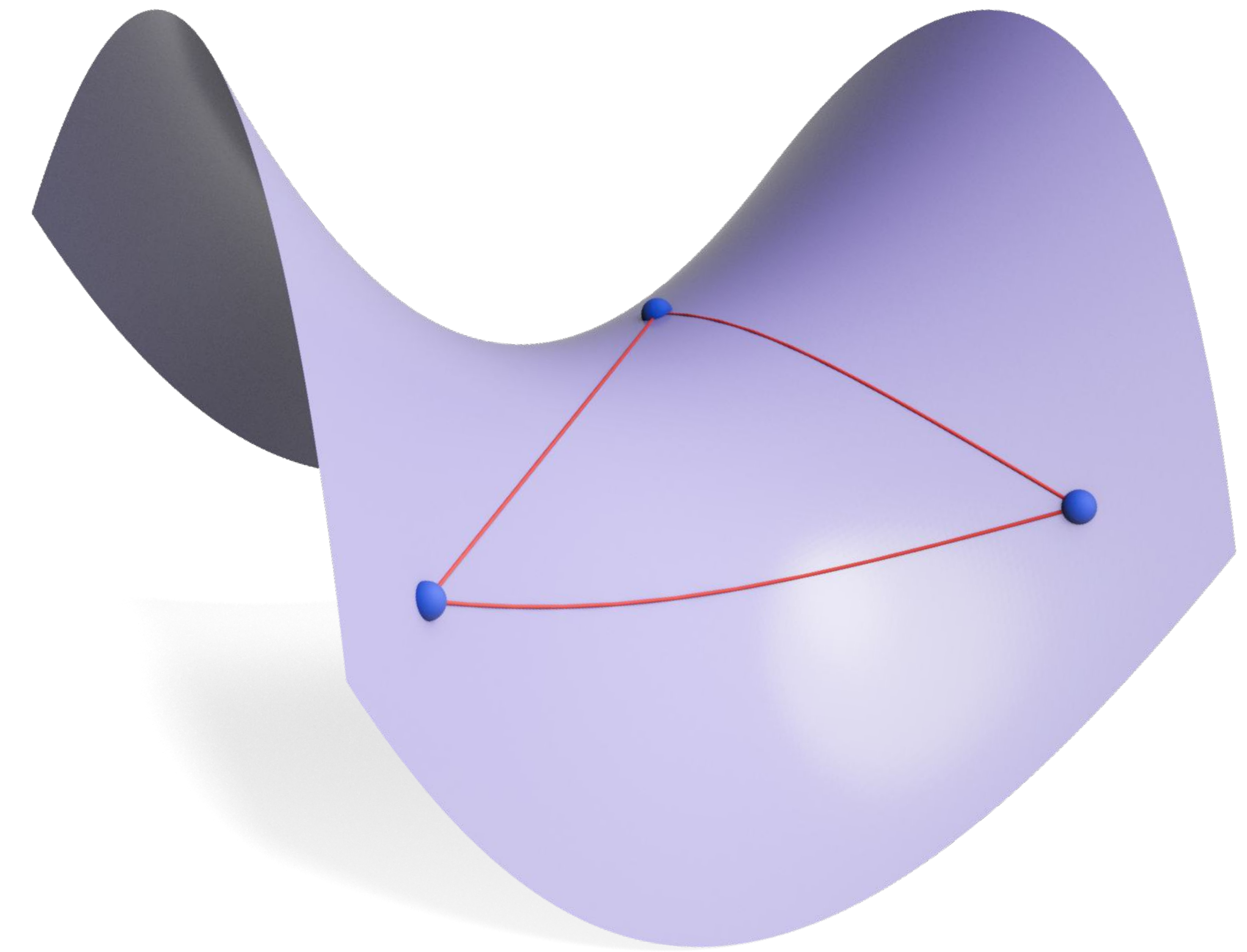}
\put(32,21){$p_1$}
\put(85,29){$p_2$}
\put(55,53){$p_3$}
\end{overpic}
\vspace{2mm}
\caption{Intrinsic Triangle on a Parametric Surface. The geodesic distance between points $p_i$ and $p_j$ is denoted by $d_g(p_i,p_j)$. The three geodesic paths forming the edges of the intrinsic triangle satisfy the triangle inequality, allowing the triangle to be developed onto a plane while preserving the lengths of its three boundary edges.
}
\label{FIG:triangle}
\end{figure}

This triangle inequality ensures that the metric space is well-defined and provides theoretical justification for the continuity of the distance field, which is essential for our subsequent processing. Most significantly, it guarantees the validity of our triangulation and the convergence of the edge-flipping operations. By operating with this induced metric structure rather than a mesh approximation, our method preserves both the computational efficiency of edge-flipping and the precise metric properties of the original surface.

\subsection{Voronoi-Guided Offset Computation}
\label{sec:Voronoi-Guided Offset Computation}

There exists a fundamental relationship between Voronoi diagrams and offset operations. Consider a set of disjoint primitives in a two-dimensional space and their generalized Voronoi diagram. When computing a global offset structure from multiple primitives, we observe a significant property: within the final offset structure, the portion originating from any specific primitive remains confined to that primitive's Voronoi cell.
This observation can be formalized as follows:

\begin{theorem}
Let $\{P_i\}$ be a set of primitives with corresponding Voronoi cells $\{V_i\}$, and let $O(P_i,d)$ denote the $d$-offset curve of primitive $P_i$ in the global offset structure. Then:
$$O(P_i,d) \subseteq V_i$$
\end{theorem}

\begin{proof} Consider two primitives $P_i$ and $P_j$ with Voronoi cells $V_i$ and $V_j$ respectively. Suppose, for contradiction, that there exists a point $x \in V_j$ that belongs to the offset curve of primitive $P_i$. By definition of the offset, $d(P_i,x) = d$ where $d$ is the offset distance. However, since $x$ lies in $V_j$, we have $d(P_j,x) < d(P_i,x) = d$, which contradicts the definition of the offset.

As illustrated in Figure~\ref{FIG:voronoi_offset}, the red point on the offset result within the Voronoi cell of the triangular primitive cannot be generated by offsetting either the circle or the square primitive. \end{proof}

\begin{figure}[h]
	\centering
\begin{overpic}
[width=0.5\linewidth]{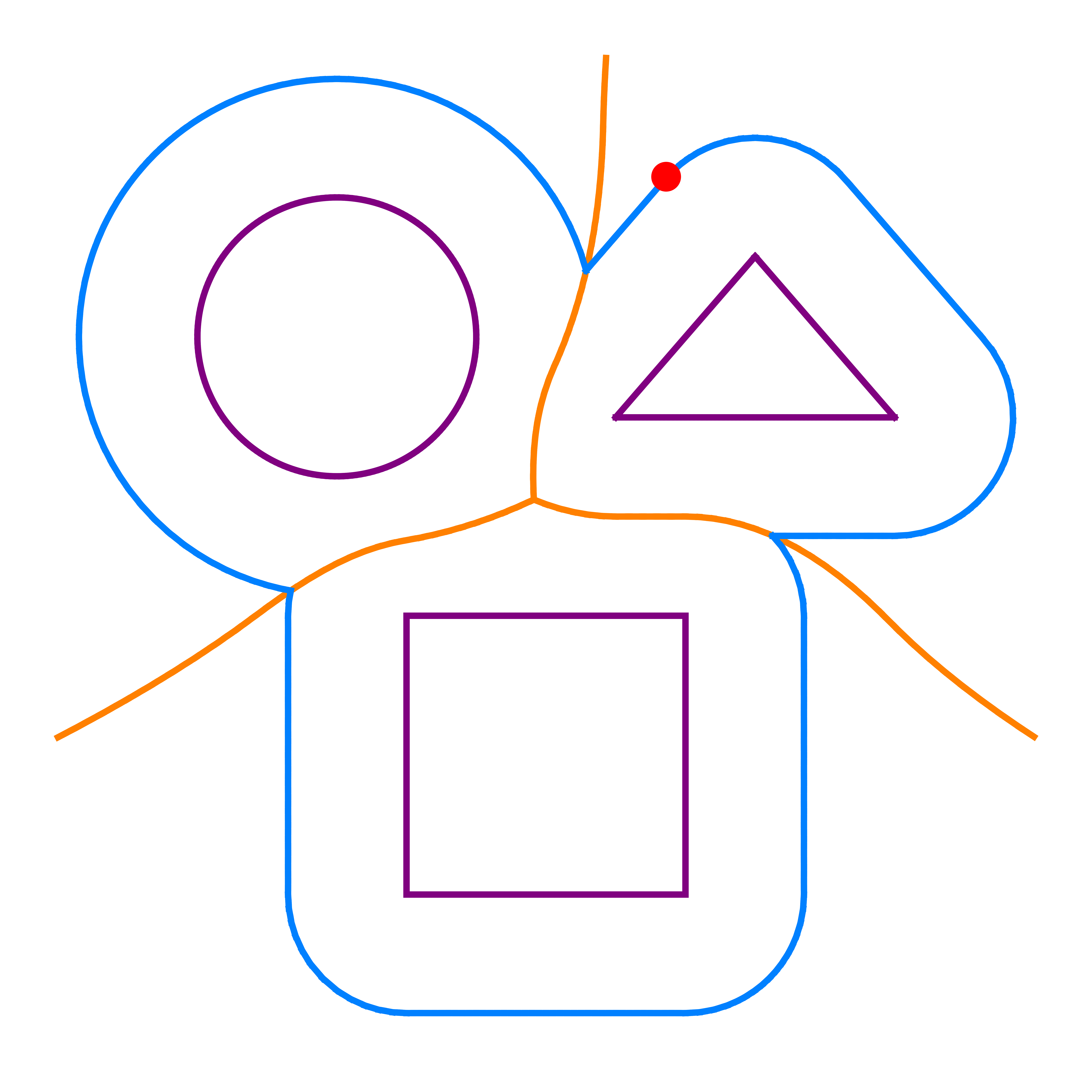}
\end{overpic}
\vspace{2mm}
\caption{
Three primitives (circle, triangle, and square, colored in purple) with their Voronoi diagram (colored in orange) and offset results (colored in blue). The offset region formed by each primitive is confined within its respective Voronoi cell.
}
\label{FIG:voronoi_offset}
\end{figure}

This theoretical foundation suggests an efficient approach for computing offsets: initially construct the Voronoi diagram of the primitives, then extract the offset structure for each primitive independently within its corresponding Voronoi cell. This divide-and-conquer strategy can substantially reduce the computational complexity of offset operations in specific scenarios. In our context, we can partition the source curve into distinct segments and compute their Voronoi diagram on the parametric space. The offset curves can subsequently be extracted independently within each Voronoi cell and amalgamated to form the complete offset structure.

\section{Implementation}
The key concept of our method is to discretize the source curve into segments, represent each segment using a point, and then extract the offset curve of each segment within its corresponding Voronoi cell.
Our algorithm is structured into three distinct phases:
\begin{enumerate}
\item Constructing Intrinsic Triangulation: Discretize the source curve into segments and represent them as points, then construct an intrinsic triangulation of the parameter space that preserves the surface's metric properties.
\item Voronoi Diagram Computation: Compute the Voronoi diagram of the representative points in the parameter space under the induced metric.
\item Offset Curve Extraction: Extract the offset curve from each Voronoi cell. Merge these individual components to obtain the final result.
\end{enumerate}

In the following sections, we present each phase in detail.

\subsection{Constructing Intrinsic Triangulation}
Our approach begins with preprocessing both the input curve and the underlying parameter space. To handle curves of arbitrary complexity, we first partition the input curve into small segments and represent each segment with a representative point. This discretization strategy preserves the curve's geometric structure while enabling efficient geodesic distance computation. With this discretized representation in hand, we proceed to construct a triangulation in the parameter space equipped with the surface-induced metric.

Consider a parametric surface defined over the domain $u \in [u_0, u_1], v \in [v_0, v_1]$. Our method commences by establishing a uniform triangulation of the parameter domain. This triangulation constitutes the fundamental infrastructure for geodesic computations, where the distance between any two points $p$ and $q$ in the parameter space is characterized as the geodesic distance between their corresponding surface points. For the initial triangulation where edges exhibit relatively minimal length, we approximate this geodesic distance between $p=(u_1,v_1)$ and $q=(u_2,v_2)$ utilizing the arc length along the surface:
\begin{equation}
L = \int_0^1 \sqrt{E(u_2-u_1)^2 + 2F(u_2-u_1)(v_2-v_1) + G(v_2-v_1)^2}  dt
\end{equation}
where $E$, $F$, and $G$ represent the coefficients of the first fundamental form. This approximation demonstrates high accuracy for infinitesimal distances and provides the initial metric structure essential for our triangulation methodology.

Furthermore, to enable geodesic distance computation between any vertex and the source curve, we incorporate the representative points of the curve segments into the parameter space triangulation through incremental insertion and triangle subdivision.

\subsection{Voronoi Diagram Computation in Parameter Space}

Voronoi diagrams have found extensive applications in computational geometry, robotics, and manufacturing due to their ability to partition space based on proximity relationships~\cite{10845125,Li2023,Meng2023,wang2024efficientnearestneighborsearch}. While numerous algorithms exist for computing Voronoi diagrams in various contexts, we adapt the SurfaceVoronoi algorithm proposed by Xin et al.~\cite{xin2022surfacevoronoi} to compute the Voronoi diagram in the parametric space, which provides a flexible framework for computing Voronoi diagrams on surfaces with arbitrary distance metrics. This approach effectively decouples the distance computation method from the Voronoi diagram construction, making it particularly suitable for our parameter space with its induced metric structure.

The algorithm consists of two main stages: the distance field propagation stage and the incremental half-plane cutting stage. In the propagation stage, we modify the original approach by replacing the distance computation with our geodesic distance calculation method described in Section~\ref{sec:Geodesic Computation via Parameter Space Triangulation}. This ensures that the propagation of distance fields throughout the parameter space accurately reflects the intrinsic geometry of the surface. Additionally, we constrain the propagation region, preventing it from extending to areas that are significantly more distant from the source curve compared to the specified offset distance.

An example of a parametric surface and its Voronoi diagram for discrete curve points in parameter space is illustrated in Figure~\ref{FIG:voronoi_disance_offset}. 
Comprehensive algorithmic details regarding the construction and computational implementation of this diagram are presented in ~\cite{xin2022surfacevoronoi}.

\begin{figure}[h]
	\centering
    \vspace{3mm}
\begin{overpic}
[width=1\linewidth]{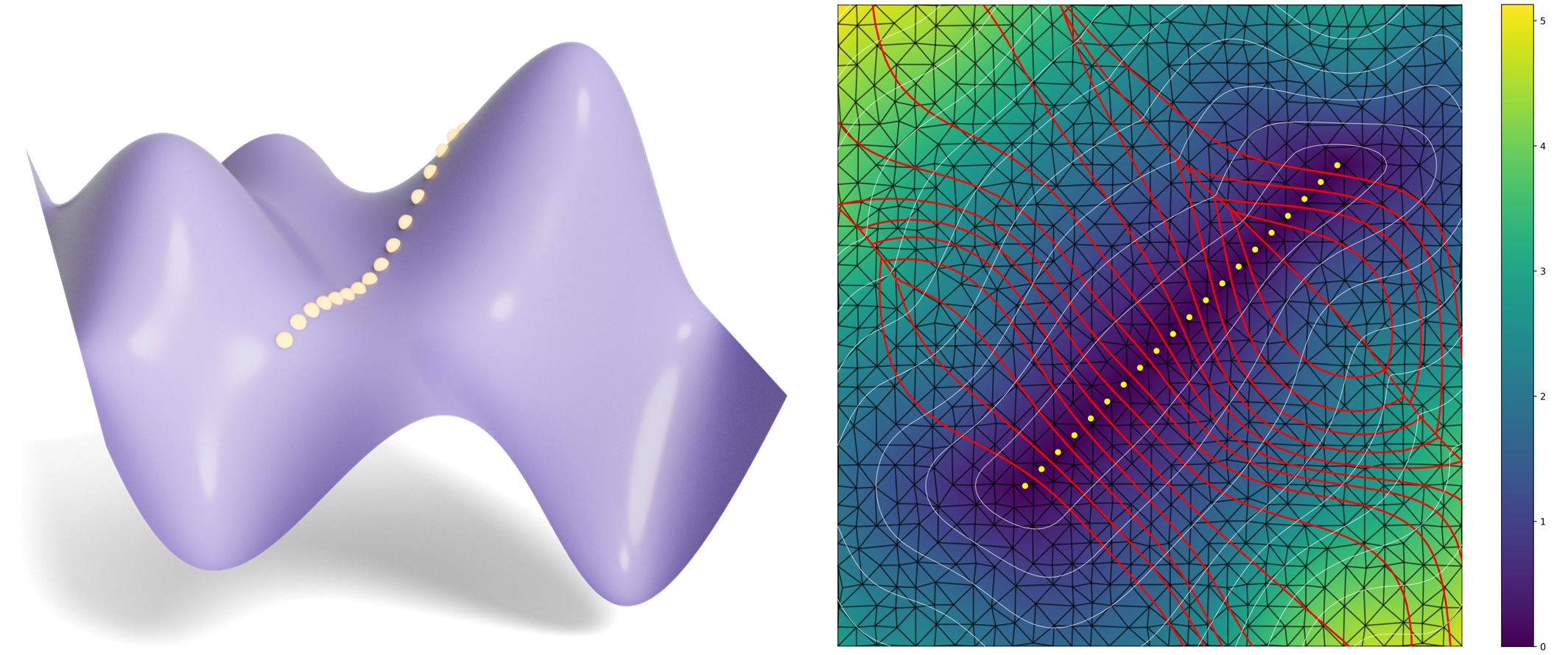}
\end{overpic}
\vspace{2mm}
\caption{
Voronoi diagram (red) of discretized curve points in parameter space under the parametric surface-induced geodesic metric. Note the substantial differences from conventional Euclidean-based 2D Voronoi diagrams. Background coloration indicates the geodesic distance field from the curve, black lines represent the parameter space triangulation, and white contours show geodesic isolines (equivalent to offset curves at varying distances).
}
\label{FIG:voronoi_disance_offset}
\end{figure}

\subsection{Offset Curve Extraction}

To enable per-cell extraction of offset curves, we integrate the Voronoi diagram with the original triangulation, creating a refined mesh structure. We modify the standard surface Voronoi algorithm to directly output an integrated triangular mesh rather than Voronoi edges. 
This modification avoids a separate re-triangulation step while ensuring the final mesh incorporates both Voronoi boundaries and original mesh edges. The resulting mesh maintains the critical property that each triangular face belongs to exactly one Voronoi cell, ensuring all points within a face are geodesically closest to that cell's site.

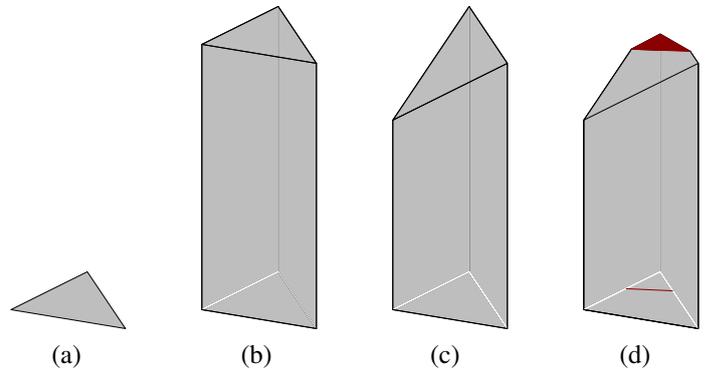
\begin{figure}[h]
	\centering
\begin{overpic}
[width=1\linewidth]{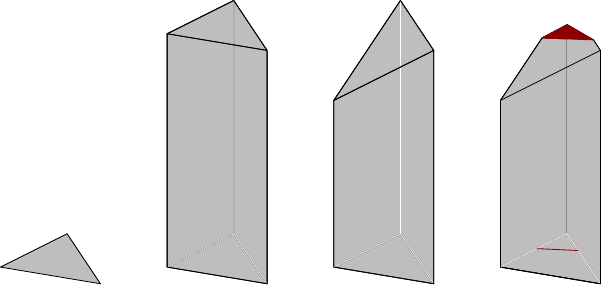}
\put(6,-5){(a)}
\put(33.5,-5){(b)}
\put(61,-5){(c)}
\put(88.5,-5){(d)}
\end{overpic}
\vspace{2mm}
\caption{
Extracting the offset curve from a triangle.
(a) Unfold the triangle onto a plane along its edges, forming an infinitely high prism with this triangle as the base.
(b) Cutting the prism with the plane determined by the distance field.
(c) Incrementally cutting the prism with planes at a constant height $d$, where $d$ equals to the offset distance. Project the intersection segments onto the triangle to obtain the offset curve.}
\label{FIG:Halfplane_Cutting}
\end{figure}

\begin{figure*}[h]
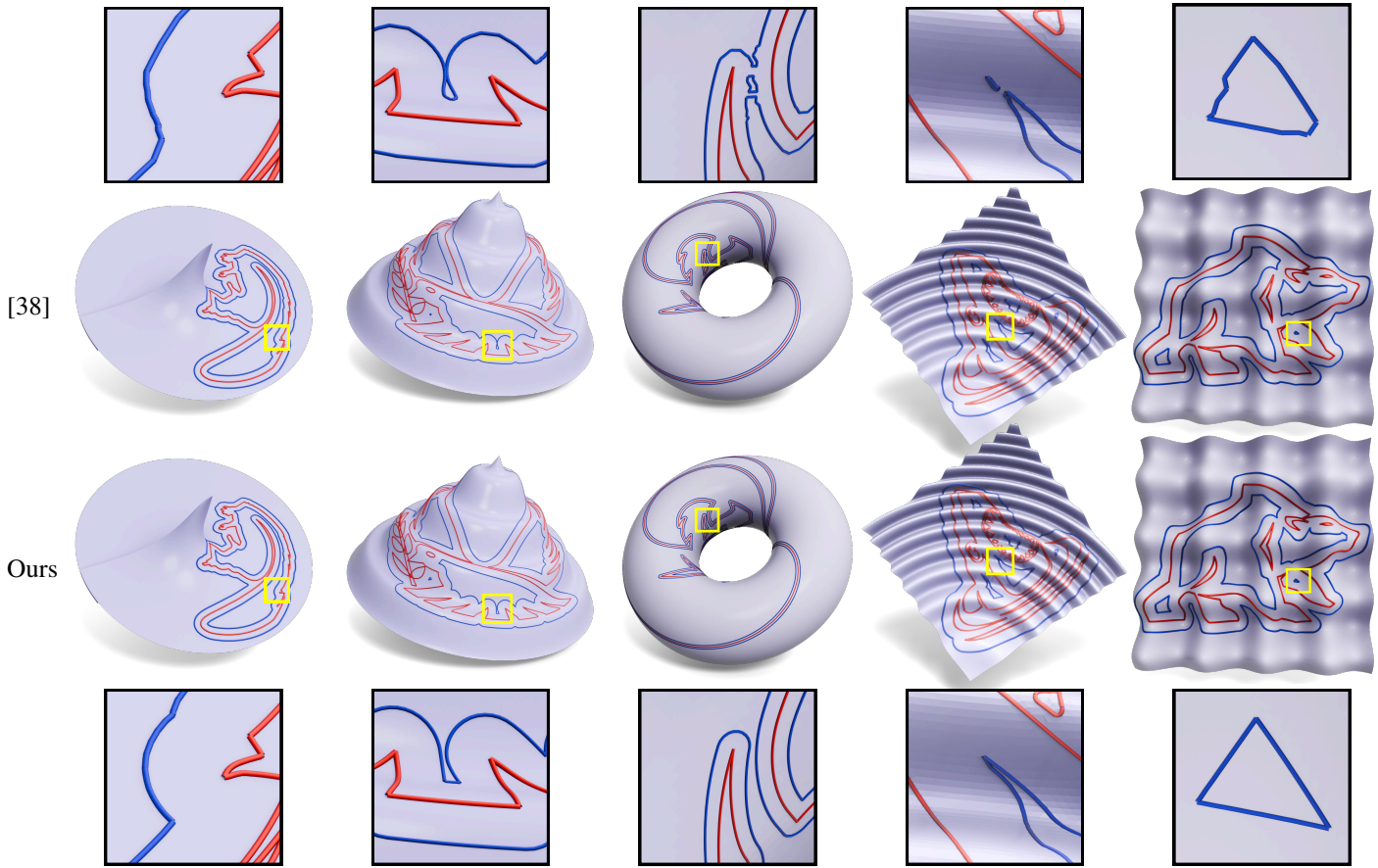

	\centering
\begin{overpic}
[width=0.95\linewidth]{imgs/comparewithxin1.pdf}
\put(-4,22){Ours}
\put(-4,42){\cite{XIN20111468}}
\end{overpic}
\vspace{2mm}
\caption{
Visual comparison of offset curve computation between our method (bottom) and ~\cite{XIN20111468} (top) on various parametric surfaces. 
Our method demonstrates more accurate alignment with theoretical offsets and better preservation of geometric features. The original source curves are shown in red, with offset curves in blue.
}
\label{FIG:compare}
\end{figure*}

With the integrated mesh structure in place, we proceed to construct a discrete distance field by computing geodesic distances at each vertex. For vertices from the original triangulation, we compute the geodesic distance from each vertex to the site of the Voronoi cell it belongs to.
For Voronoi vertices, which are newly introduced points at the intersection of multiple Voronoi cells, we compute geodesic distances to each adjacent Voronoi cell's site on the source curve.While these vertices should theoretically be equidistant from adjacent sites, our approximate Voronoi diagram computation may yield slightly different values. For consistency, we assign each Voronoi vertex the minimum geodesic distance from among its adjacent cells' sites, ensuring every vertex in our modified mesh has a well-defined geodesic distance value.

Since each triangle in our refined mesh belongs to exactly one Voronoi cell and has well-defined distance values at its vertices, we can extract offset curves independently for each triangle.
By linearly approximating the distance field within each triangle, we transform a complex geometric problem into a simple plane-cutting operation. This localization enables parallel processing and eliminates the need for global self-intersection handling. With the distance field established at every vertex, the extraction procedure for each triangle is straightforward:

\begin{enumerate} 
\item We isometrically develop each triangle onto a plane and construct a vertical prism that extends infinitely in both directions.
\item We create a sloped plane where the height at each vertex equals its geodesic distance to the source curve. This plane, representing the linearly interpolated distance field, intersects and cuts the prism, creating a half-space envelope.
\item By further cutting this remaining portion with a horizontal plane at height equal to the specified offset distance, we obtain the 3D intersection curve which, when projected vertically back onto the 2D triangle domain, forms the local offset curve segment.
\end{enumerate}

Figure~\ref{FIG:Halfplane_Cutting} illustrates this extraction process. By combining these independently extracted segments, we form the complete offset curve that naturally respects the topology of the distance field.

\section{Evaluation}
\subsection{Experimental setting}
We implemented our algorithm in C++ on a platform with a 2.8 GHz AMD 5050X processor running Windows 10. We modified the code from~\cite{10.1145/3414685.3417839} to meet our requirements for geodesic computation on parametric surfaces.

\begin{figure*}[h]
	\centering
\begin{overpic}
[width=0.98\linewidth]{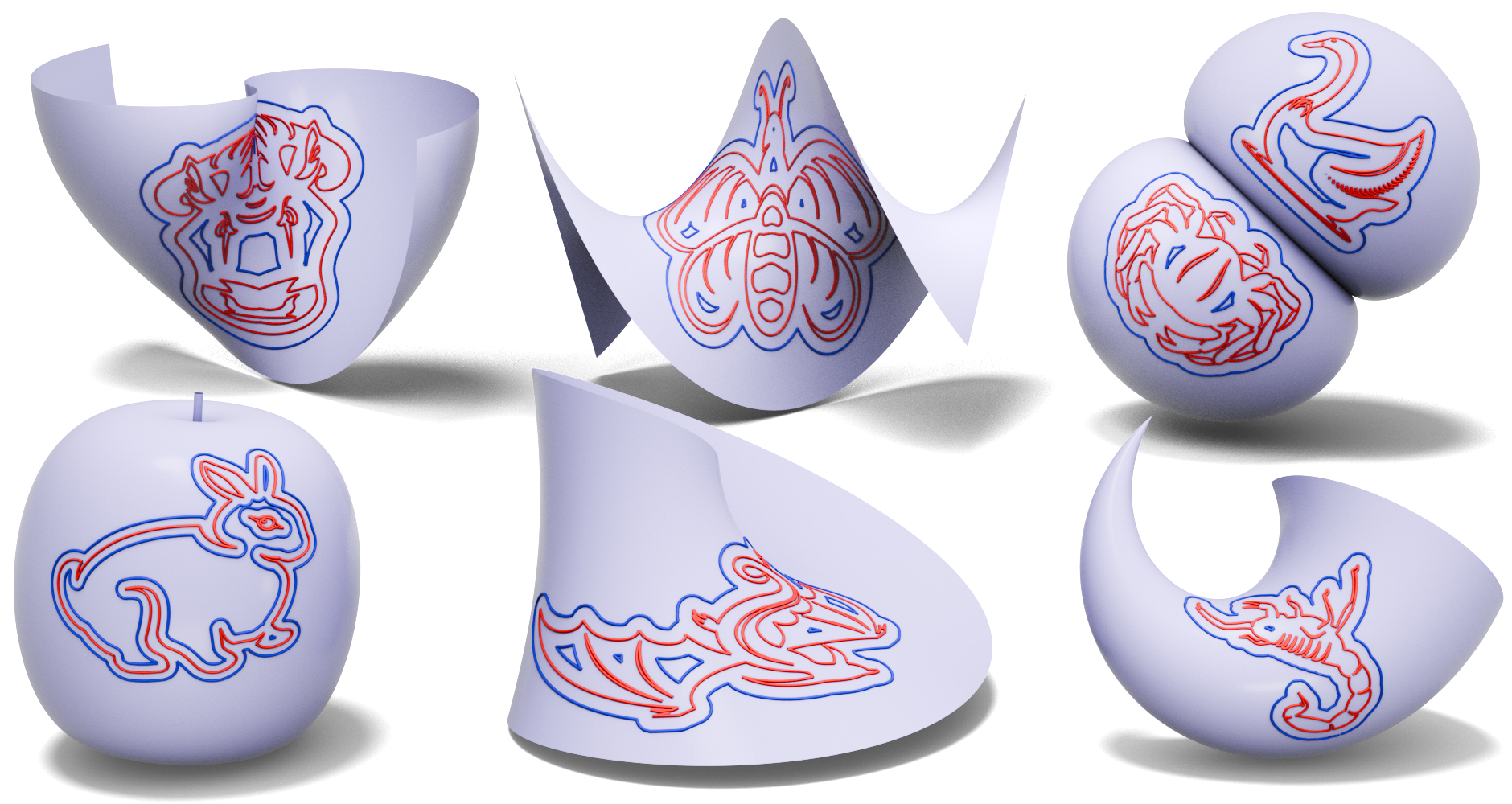}
\end{overpic}
\vspace{2mm}
\caption{
Gallery of offset curves computed using our method on various parametric surfaces. 
For each model, the source curve is shown in red, with offset curves in blue. 
}
\label{FIG:gallery1}
\end{figure*}

\begin{figure}[h]
	\centering
\begin{overpic}
[width=1\linewidth]{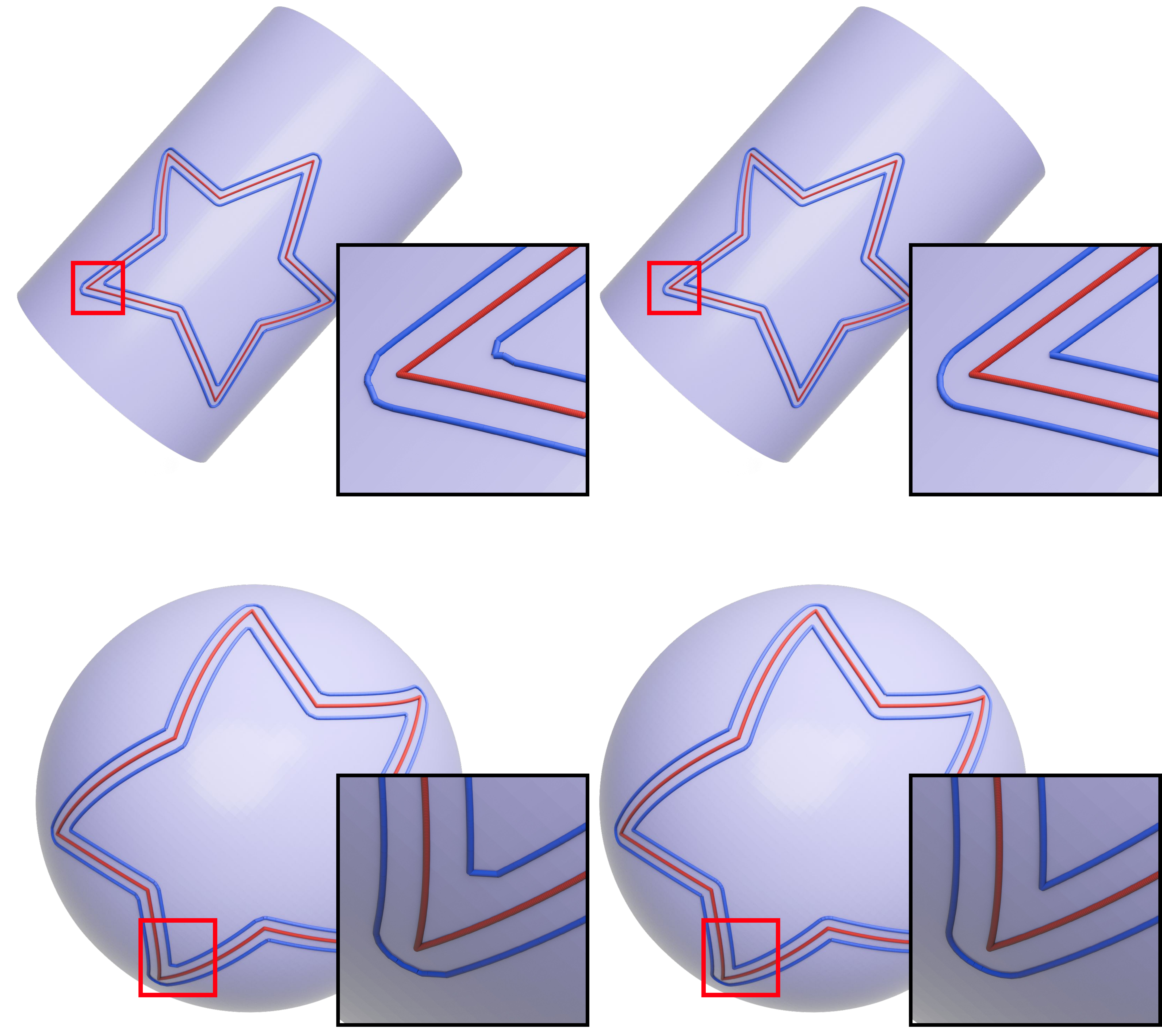}
\put(20,-5){\cite{XIN20111468}}
\put(69,-5){Ours}
\end{overpic}
\vspace{2mm}
\caption{
Visual comparison of offset curves on cylindrical (top) and spherical (bottom) surfaces. Left: results by ~\cite{XIN20111468}; right: our results. Original curves are shown in red, with offset curves in blue.
}
\label{FIG:hausdoff}
\end{figure}

\subsection{Comparison with the State of the Art}
In this paper, we compare our method with Xin et al.\cite{XIN20111468}, which is designed for computing offset curves on triangulated mesh surfaces. We do not compare with analytical methods for parametric surfaces, as these are typically applicable only to specific surface types and often lack publicly available implementations.
Our method requires intrinsic triangulation of the parametric space, while Xin et al.\cite{XIN20111468} requires extrinsic triangulation of the parametric surface.
For a fair comparison, we ensure that the number of triangulations in both methods remains consistent.

\paragraph{Visual comparison}
We compare our method with Xin et al.~\cite{XIN20111468} on five complex models. The experimental results are shown in Figure~\ref{FIG:compare}. When the offset distance exceeds the curvature radius of the input curve, distinct feature angles, or cusps, appear in the offset curve. The ability to generate cusps is a key indicator of the quality of an offset result. As seen in the figure, our method clearly produces cusps, whereas Xin et al.'s method does not. This is because our approach partitions the curve into different regions, computes the offsets, and extracts them individually. Each Voronoi cell captures the local characteristics of the original curve, leading to a more accurate approximation of the distance field. In contrast, methods that extract offsets only within triangles lack this region-specific information and thus cannot achieve the same level of accuracy. Figure~\ref{FIG:gallery1} displays additional results obtained using our method.

\paragraph{Quantitative comparison}
Traditional surfaces cannot support precise geodesic distance calculations. To quantitatively compare our method with \cite{XIN20111468}, we selected spherical and cylindrical surfaces for testing. Using identical input curves and offset distances, we compared our method with \cite{XIN20111468}. The results are shown in Figure\ref{FIG:hausdoff}. We sampled 100K points from both the parameter curve and the computed offset curve, and calculated the one-directional Hausdorff distance and Chamfer distance. As shown in Table~\ref{tab:quantitative}, our method demonstrates superior accuracy across these metrics.

\begin{table}[]
\centering
\renewcommand{\arraystretch}{1.2} 
\setlength{\tabcolsep}{5pt}  
\resizebox{\linewidth}{!}{
\begin{tabular}{c|cc|cc}
\toprule
\multirow{2}{*}{\textbf{Surface}} & \multicolumn{2}{c|}{\textbf{\cite{XIN20111468}}} & \multicolumn{2}{c}{\textbf{Ours}} \\ 
\cmidrule(lr){2-5}
 & HD ($\times 10^{-3}$) & CD ($\times 10^{-3}$) & HD ($\times 10^{-3}$) & CD ($\times 10^{-3}$) \\ 
\midrule
Sphere      & 7.59 & 2.11 & 0.71 & 0.29 \\ 
Cylinder    & 8.33 & 0.15 & 0.76 & 0.02 \\
\bottomrule
\end{tabular}
}
\caption{Quantitative comparison of offset curve accuracy using one-directional Hausdorff Distance (HD) and Chamfer Distance (CD).}
\label{tab:quantitative}
\end{table}

\paragraph{Performance comparison}
Table~\ref{tab:time_comparison} presents the computational times for our method and \cite{XIN20111468} on the models shown in Figure~\ref{FIG:compare}. Our method demonstrates a clear advantage in terms of computational efficiency. This performance gain is primarily due to our approach for computing geodesic distances, which directly utilizes intrinsic edge-flipping operations. Compared to traditional exact geodesic computations, our method achieves significantly faster processing times.

\begin{table}[h]
    \centering
    \renewcommand{\arraystretch}{1.2} 
    \setlength{\tabcolsep}{11.5pt}  
    \resizebox{\linewidth}{!}{
        \begin{tabular}{c|c|c|c|c} 
        \toprule
        \multirow{2}{*}{\textbf{Surface}} & \multirow{2}{*}{\textbf{\#V}} & \multirow{2}{*}{\textbf{\#F}} & \multicolumn{2}{c}{\textbf{Time (s)}} \\
        \cmidrule(lr){4-5}
        & & & \textbf{Ours} & \textbf{\cite{XIN20111468}} \\
        \midrule
        Bell Curve & 25k & 50k & 1.173 & 1.926 \\
        Spiral Paraboloid & 26k & 51k & 0.981 & 2.184 \\
        Torus & 50k & 99k & 1.744 & 4.623 \\
        Circular Wave & 34k & 66k & 1.291 & 2.782 \\
        Bivariate Sine Wave & 40k & 80k & 2.116 & 4.741 \\
        \bottomrule
    \end{tabular}
    }
    \caption{Runtime performance comparison between our method and \cite{XIN20111468} on models shown in Figure~\ref{FIG:compare}. \#V and \#F represent the number of vertices and faces in the triangulation, respectively. In this experiment, the curve was discretized into 2K segments.}
    \label{tab:time_comparison}
\end{table}

\subsection{Time Complexity Analysis}
To evaluate the algorithm's computational efficiency, we conducted two experiments examining how discrete segment count and offset distance affect processing time.

\paragraph{Number of Discrete Segments}
We discretized the input curve into varying numbers of segments while maintaining a constant offset distance of 0.3 and 50,000 intrinsic triangles. Figure~\ref{FIG:time-points} illustrates the results. The computation time increases linearly with the number of sampled points, demonstrating the algorithm's scalability even as input complexity grows.

\begin{figure}[h]
	\centering
    \vspace{3mm}
\begin{overpic}
[width=1\linewidth]{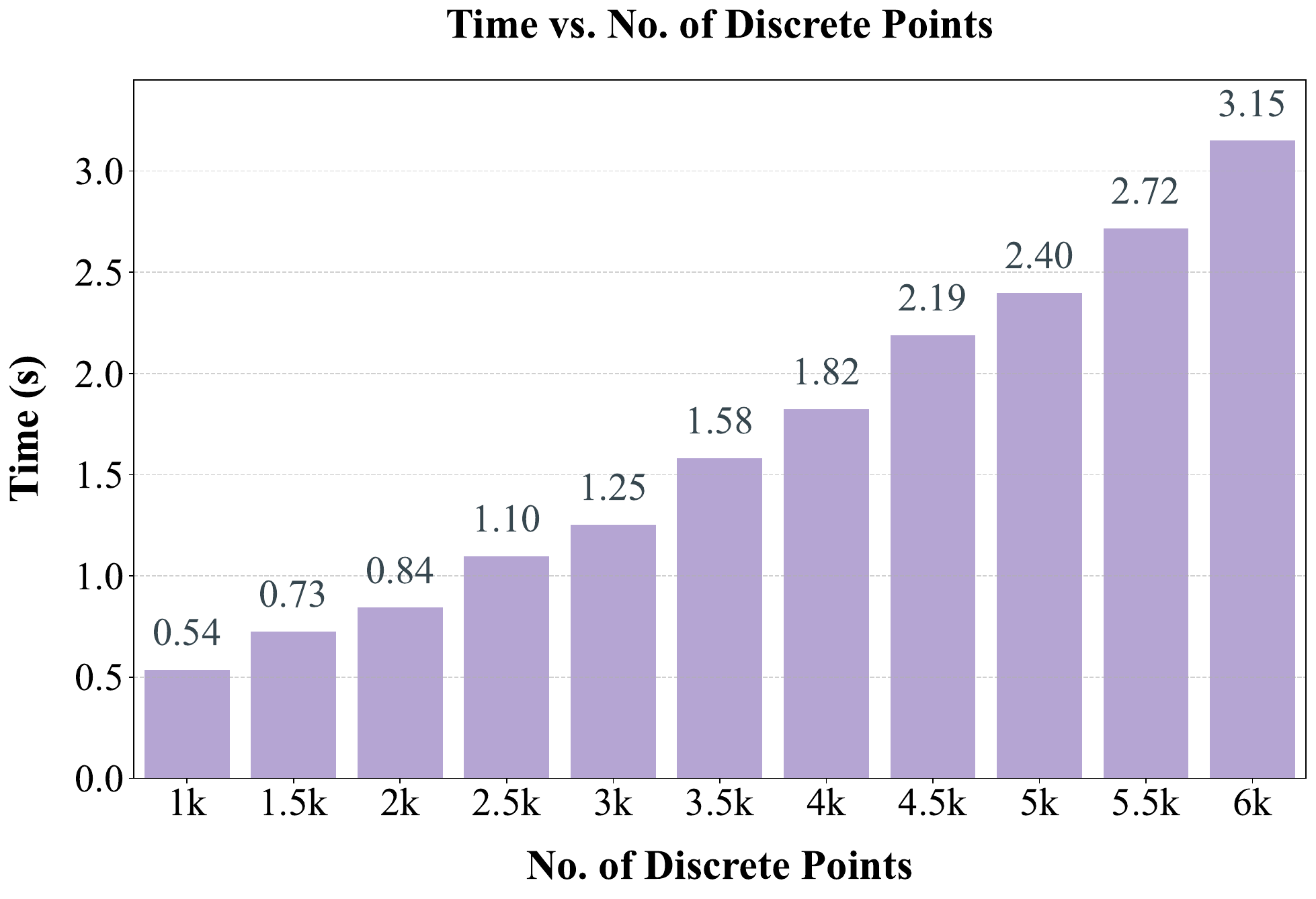}
\end{overpic}
\caption{
Computation time statistics with respect to discrete points. The results clearly show that computation time scales linearly with the number of discrete points.
}
\label{FIG:time-points}
\end{figure}

\begin{figure}[h]
	\centering
    \vspace{3mm}
\begin{overpic}
[width=1\linewidth]{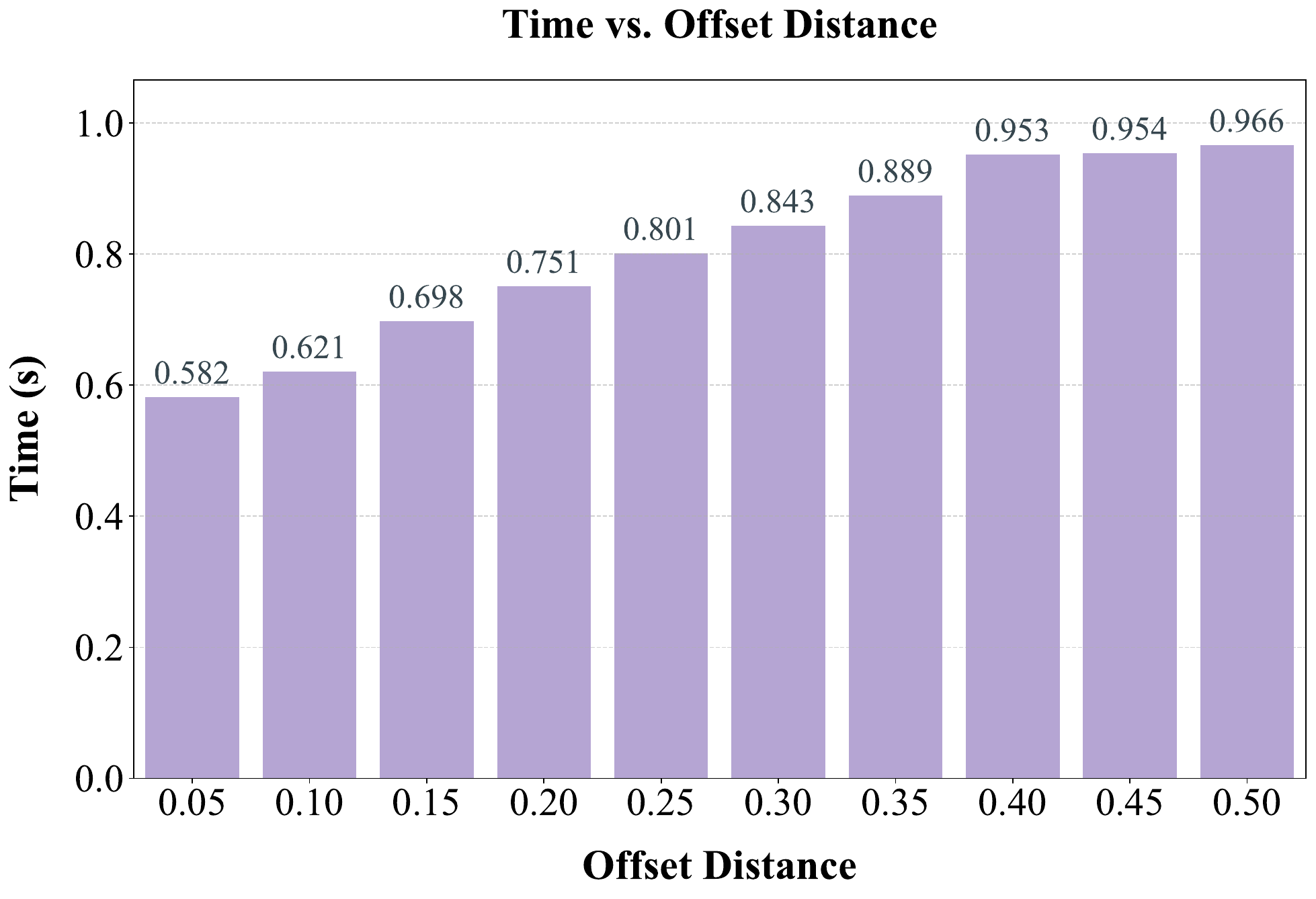}
\end{overpic}
\caption{
Time statistics with respect to offset distance.
The computation time increases with larger offset distances, owing to their impact on the propagation region of the distance field.
}
\label{FIG:time-offset}
\end{figure}

\paragraph{Offset Distance}
Our second experiment investigated the relationship between computational time and offset distance using 50,000 intrinsic triangles and 2,000 segments. Figure~\ref{FIG:time-offset} displays these results. Processing time increases gradually with larger offset distances because the distance field's propagation stage must cover more extensive surface regions. Nevertheless, the algorithm maintains efficient performance across varying offset distances, exhibiting robust scalability with increasing task complexity.

\subsection{Ablation Study}
The accuracy of our results is primarily determined by the number of discrete segments, with minimal dependency on the number of intrinsic triangles. This is because after computing the Voronoi diagram, we embed it into the final triangulation, effectively refining the input triangular mesh. Consequently, the initial intrinsic triangulation has little impact on the final result. Therefore, we focused our analysis solely on how the number of segments affects our results. Additionally, we compared the accuracy of geodesic distance computed using our intrinsic triangle flipping method against those obtained by discretizing the parametric surface into a triangular mesh on a spherical surface.

\subsubsection{Variable Discretization Precision}
Figure~\ref{FIG:ResultsunderDifferentPoints} demonstrates our results under different discretization precisions of the input curve. Visual comparison reveals that at lower discretization densities, the offset results clearly exhibit a segmented appearance composed of multiple distinct arc segments. This occurs because fewer discrete points are insufficient to accurately represent the source curve. As sampling density increases, results become increasingly accurate. When our sampling number reaches approximately 200, we achieve visually satisfactory results.

\begin{figure}[h]
	\centering
\begin{overpic}
[width=0.98\linewidth]{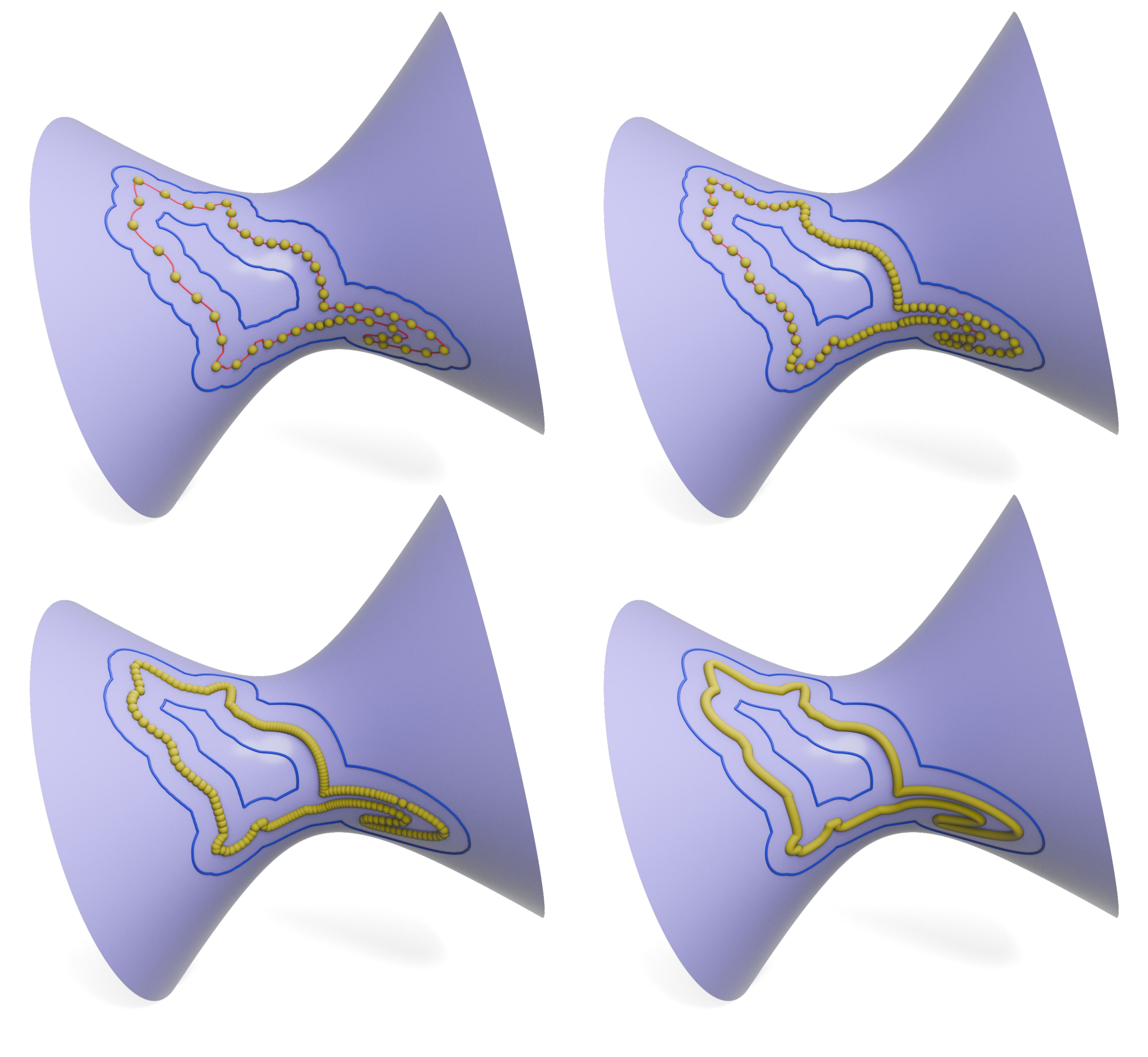}
\end{overpic}
\vspace{2mm}
\caption{
Offset results under different curve discretization densities. From left to right, top to bottom: results obtained with 50, 100, 200, and 2000 sampling points. Lower discretization densities produce visible segmentation artifacts. The results become visually satisfactory when the curve discretization density exceeds 200. The source curve is shown in red, with the offset curve in blue.
}
\label{FIG:ResultsunderDifferentPoints}
\end{figure}

\begin{table}[h]
\centering
\renewcommand{\arraystretch}{1.2} 
\setlength{\tabcolsep}{5pt}  
\resizebox{\linewidth}{!}{
\begin{tabular}{c|ccccc} 
\toprule
\textbf{\#F}   & \textbf{80}     & \textbf{320}    & \textbf{1280}   & \textbf{5120}   & \textbf{20480}   \\ 
\midrule
Ours   & 1.25\% & 0.43\% & 0.12\% & 0.03\% & 0.008\% \\
Remeshing & 2.67\% & 0.81\% & 0.23\% & 0.06\% & 0.02\% \\ 
\bottomrule
\end{tabular}
}
\caption{ Comparison of Geodesic Distance Accuracy: Deviation rates (\%) from analytical solutions for our intrinsic triangulation method versus the traditional remeshing approach across varying triangulation densities. Lower values denote a more accurate approximation of the exact geodesic distances.}
\label{table:geodesic}
\end{table}

\subsubsection{Variable Triangulation Density Analysis}

Figure~\ref{FIG:ResultsunderDifferenttriangles} presents the offset curves generated by our algorithm under different mesh resolutions. 
The computed offset curves show consistency across varying triangulation densities, with close-up windows highlighting areas of minor deviation. 
Our approach yields satisfactory results even at lower resolution settings, 
which can be attributed to our Voronoi cell-based extraction methodology. 
By embedding the Voronoi diagram into the triangular mesh and extracting offset curves on a per-cell basis, we effectively mitigate the influence of global mesh density. 
The observed geometric discrepancies are primarily attributable to computational inaccuracies in the Voronoi diagram construction induced by triangulation coarseness, notwithstanding the adequate sampling density of the source curve itself.

\begin{figure}[h]
	\centering
\begin{overpic}
[width=1\linewidth]{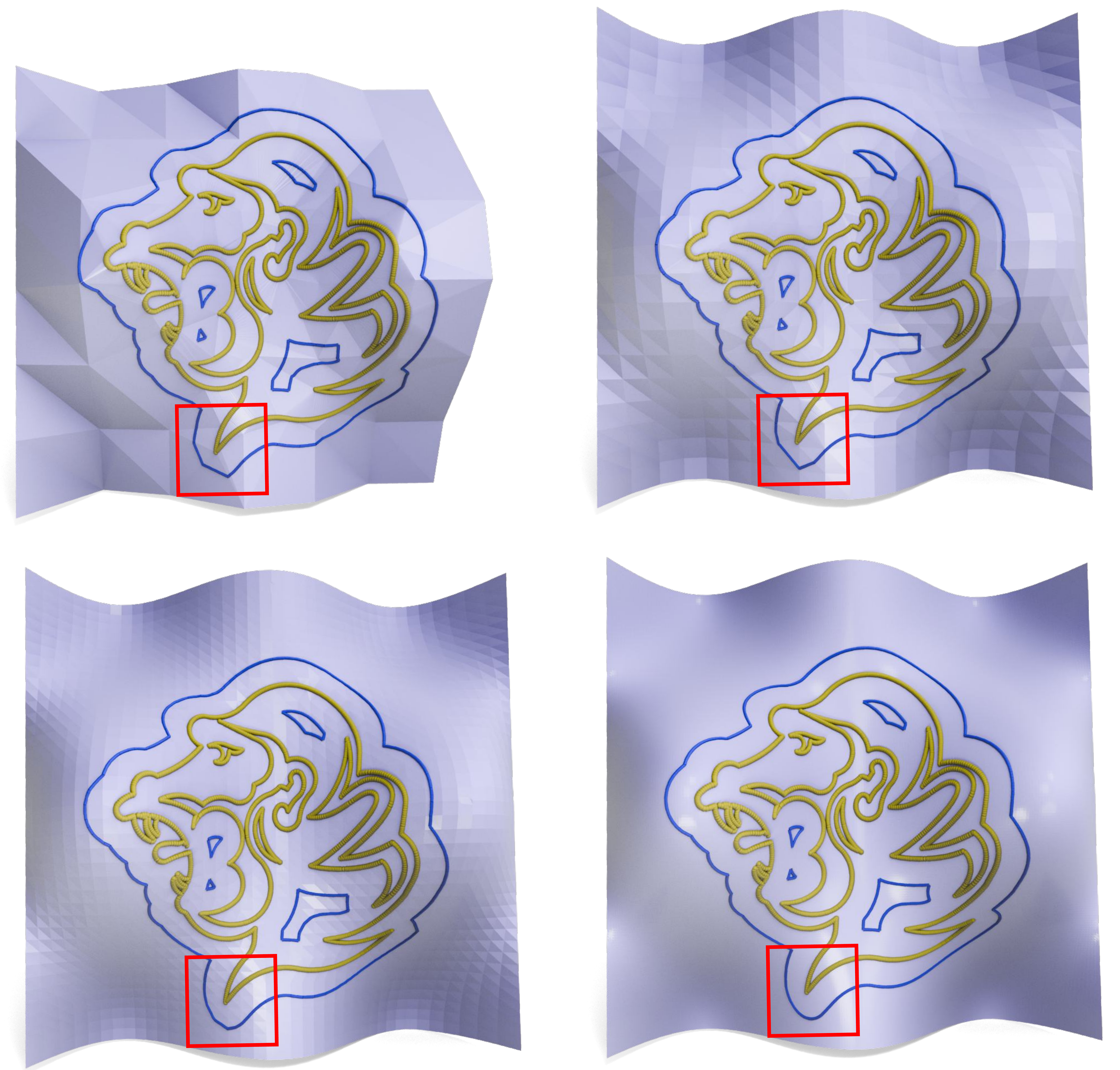}
\end{overpic}
\vspace{2mm}
\caption{
Offset curve results generated under different surface mesh resolutions. From left to right, top to bottom: results computed on meshes consisting of 72 triangles, 800 triangles, 3200 triangles, and 80000 triangles, respectively, illustrating the algorithm's consistency across multiple discretization levels.
}
\label{FIG:ResultsunderDifferenttriangles}
\end{figure}

\begin{figure}[h]
	\centering
    \vspace{3mm}
\begin{overpic}
[width=1\linewidth]{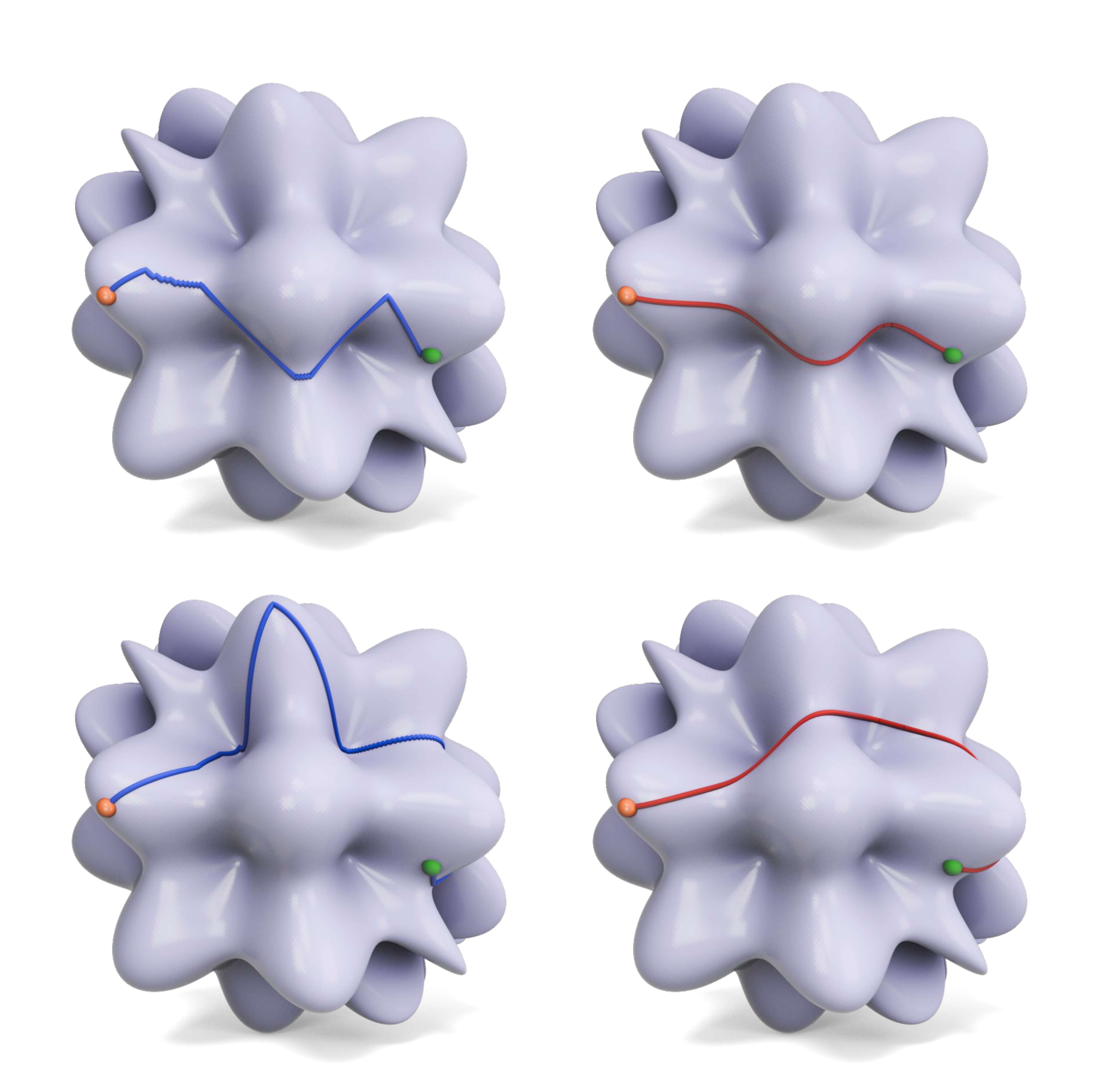}
\put(22,49){(a)}
\put(74,49){(b)}
\put(22,-2){(c)}
\put(74,-2){(d)}

\end{overpic}
\vspace{2mm}
\caption{
Geodesic paths computed under different initialization strategies (mapped from parameter space to the surface for improved visualization).
(a) Dijkstra algorithm initialization path; (b) Optimized geodesic path derived from Dijkstra initialization; (c) Suboptimal initialization path; (d) Optimized geodesic path derived from suboptimal initialization. Orange and green dots denote the start and end points, respectively.
}
\label{FIG:singlegeodesicfailed}
\end{figure}

\begin{figure}[h]
	\centering
    \vspace{3mm}
\begin{overpic}
[width=1\linewidth]{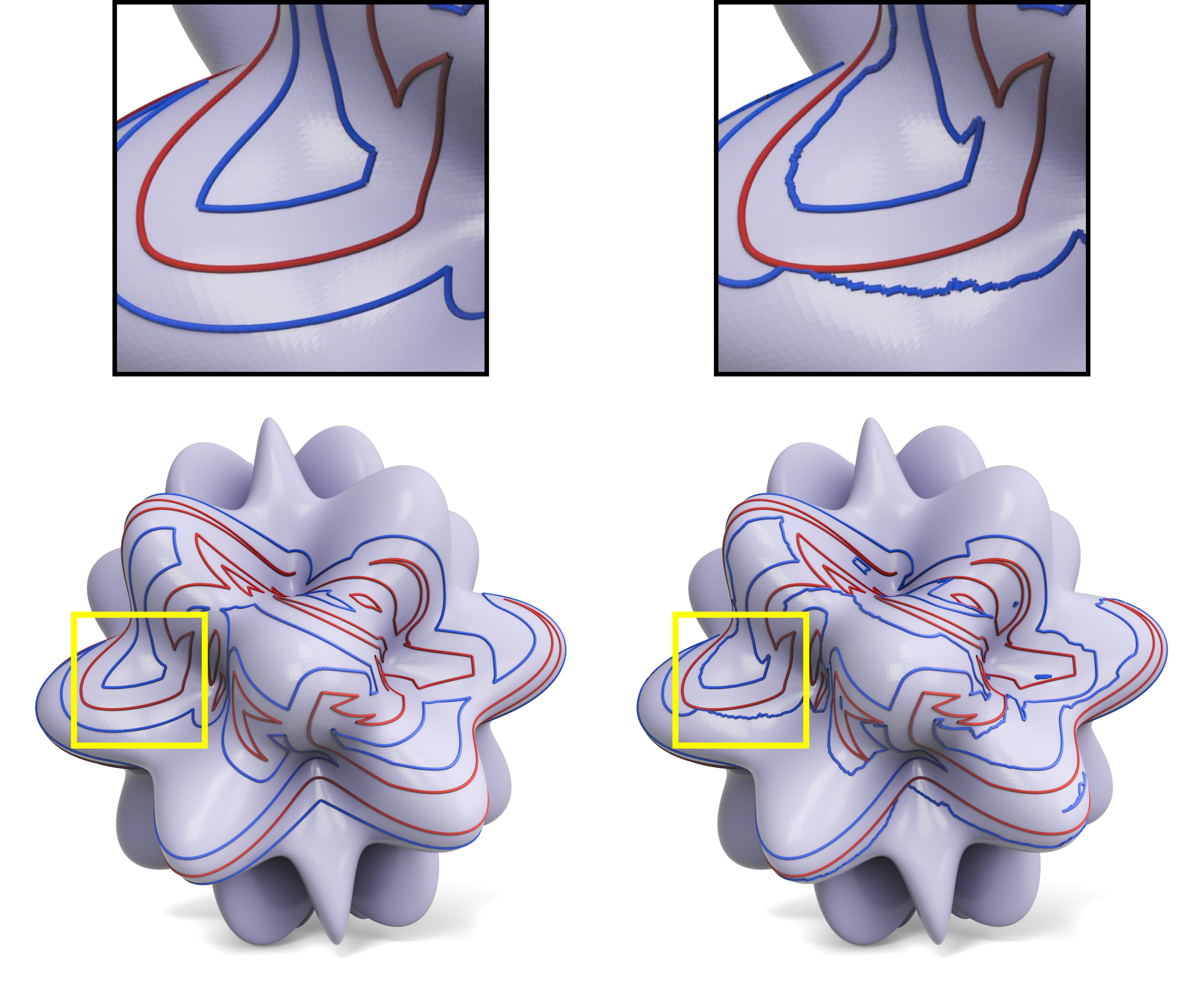}
\put(23,-5){(a)}
\put(74,-5){(b)}

\end{overpic}
\vspace{2mm}
\caption{
Comparative visualization of curve offset results based on different initialization methods. (a) Offset curves obtained using Dijkstra algorithm initialization; (b) Offset curves generated using suboptimal initialization strategy.
}
\label{FIG:allgeodesicfailed}
\end{figure}

\begin{figure}[h]
	\centering
    \vspace{3mm}
\begin{overpic}
[width=1\linewidth]{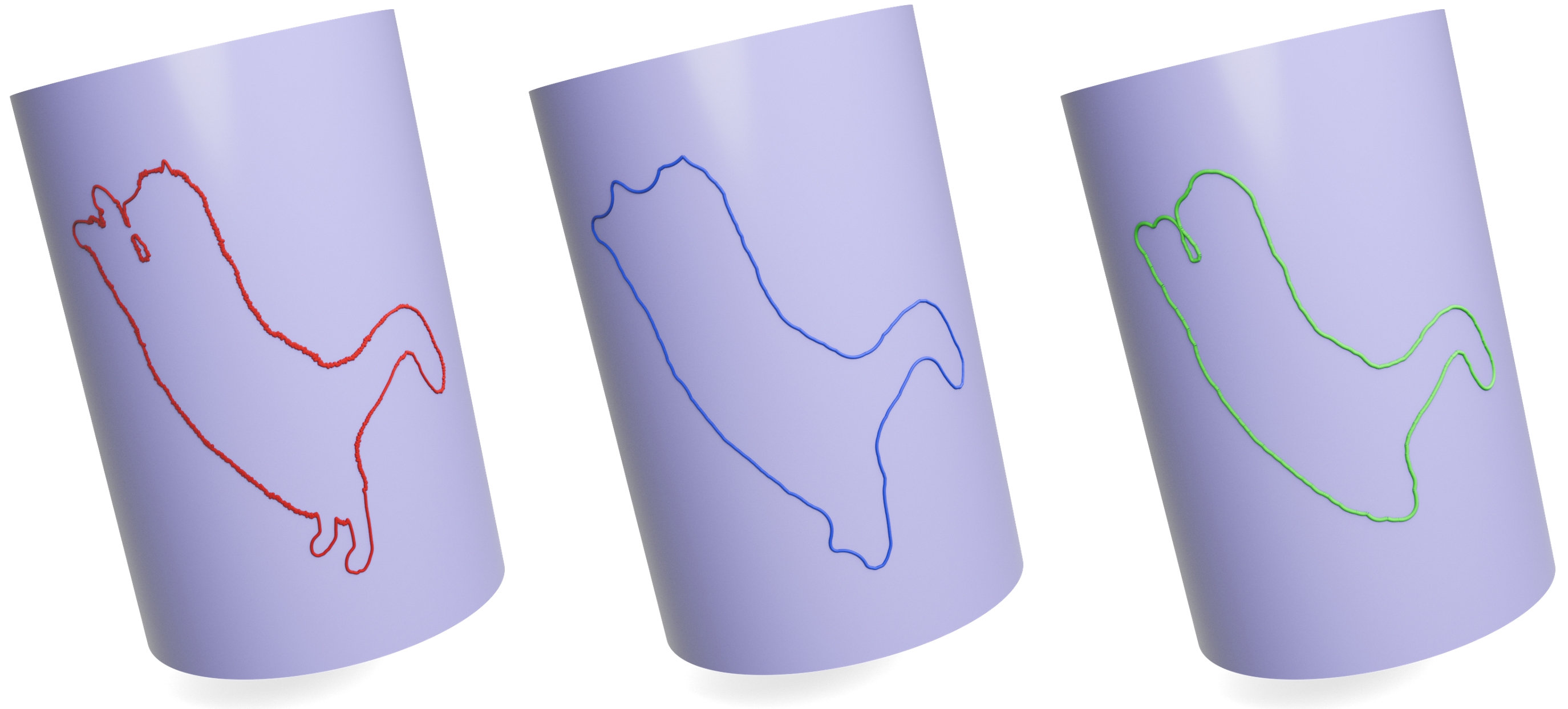}
\put(16,-4){Input}
\put(49,-4){Closing}
\put(83,-4){Opening}

\end{overpic}
\vspace{2mm}
\caption{
Morphological operations on a parametric surface. Left: original input curve (red); middle: result of the closing operation (dilation followed by erosion) shown in blue; right: result of the opening operation (erosion followed by dilation) shown in green. In the closing result, small holes are filled, while in the opening result, small structures are removed, leaving larger structures intact.
}
\label{FIG:openingAndClosing}
\end{figure}

\subsubsection{Intrinsic vs. Extrinsic Geodesic Computation}
Although our method computes geodesic distances by flipping intrinsic triangulations in the parameter domain, it is not exact since the initial intrinsic triangulation is not completely developable. To verify the accuracy of the geodesic distance, we compared it with the method of discretizing parametric surfaces into triangular meshes. For precise geodesic distance testing, we selected a sphere as the test surface due to the availability of analytical solutions between any two points on a sphere. We discretized the sphere at different resolutions, maintaining consistent face counts between both triangulation methods. We validated both methods against analytical solutions by randomly selecting start and end points, calculating geodesic distances 10,000 times, and computing the average deviation. Results in Table~\ref{table:geodesic} show our method achieves approximately twice the accuracy of the mesh-based approach.

\subsubsection{Failure Case}
Our method employs an intrinsic triangulation flipping strategy to compute geodesic distances, which, while efficient, cannot guarantee global optimality in all scenarios. The quality of the computed geodesic path depends significantly on the initialization path from which the flipping procedure begins. Throughout our experimental evaluation, we utilized Dijkstra's algorithm to generate initialization paths, and no errors were observed in any of the test cases presented in this paper.

To systematically investigate the limitations of our approach, we deliberately introduced a suboptimal initialization strategy: 
when constructing the initialization path between vertices $t$ and $s$, we inserted multiple randomly selected intermediate vertices, calculated separate Dijkstra paths between consecutive vertex pairs, then concatenated these paths to form a highly suboptimal initialization path that deliberately deviates from the expected geodesic trajectory.
Figure~\ref{FIG:singlegeodesicfailed} juxtaposes the geodesic paths obtained under both initialization strategies. The results clearly demonstrate that with poorly constructed initialization paths, the flipping procedure converges to suboptimal geodesic paths.

Furthermore, Figure~\ref{FIG:allgeodesicfailed} illustrates the impact of initialization strategy on the final offset curves. With suboptimal initialization, the computed offset curves exhibit inaccuracies.  As shown in the figure, certain points that are geometrically proximate on the surface can be calculated to have disproportionately large geodesic distances, causing the offset curve to appear much closer to the original curve than it should be in those instances.

A fundamental property of geodesic distance fields is that their gradient norm equals 1. Therefore, if we observe that the computed distance field gradient exceeds 1, it clearly indicates an error in the geodesic distance computation. Specifically, for a triangle $\triangle v_1v_2v_3$ with computed geodesic distances $d_1$, $d_2$, and $d_3$ respectively, if $|d_2 - d_1| > L(v_1,v_2)$, where $L(v_1,v_2)$ is the distance between the two points under the induced metric, then the directional gradient of the distance field along edge $v_1v_2$ exceeds 1, which is impossible and indicates computational errors in geodesic distances.
The same principle applies to the other two edges of the triangle. This strategy should detect most shortest path computation failures, as when geodesic distance computation errors occur, they are often accompanied by regions where geometrically adjacent points have vastly different Dijkstra initialization paths, leading to significantly different computed distances and resulting in substantial errors similar to the situation illustrated in Figure~\ref{FIG:singlegeodesicfailed}. We must acknowledge that this approach cannot detect 100\% of all errors.

\subsection{Opening {\em\&} Closing}
Morphological operations are essential techniques in image processing that can be extended to curves on parametric surfaces. Among these operations, opening and closing are particularly effective for smoothing, noise removal, and feature preservation. Opening is defined as an erosion followed by a dilation, whereas closing is a dilation followed by an erosion. These operations can be efficiently implemented using our offset curve computation method.

Our approach facilitates these morphological operations on parametric surfaces by computing inward and outward offsets at specified distances. Figure~\ref{FIG:openingAndClosing} illustrates an example of opening and closing operations applied to a curve on a parametric surface. As observed, in the closing result, small holes are filled, while in the opening result, small structures are removed, leaving larger structures intact.

\section{Conclusion, Limitations and Future Work}
In this paper, we introduced a novel approach for computing curve offsets on parametric surfaces. Our method is based on two key observations. First, geodesic computation on parametric surfaces is an intrinsic property that remains independent of spatial embedding. By constructing an intrinsic triangulation of the parameter space and equipping it with the surface-induced metric, we can efficiently compute geodesic distances between arbitrary points through intrinsic triangle flipping operations. Second, the offset curve of different primitives is confined within their corresponding Voronoi cells. This property allows us to discretize the source curve, compute its Voronoi diagram, and then extract offsets locally within each cell. Compared to traditional methods, our approach demonstrates superior efficiency and accuracy.

Despite these advantages, our approach has certain limitations. While the initial intrinsic triangulation satisfies the triangle inequality, strictly speaking, the intrinsic triangulation is not necessarily unfoldable. This introduces some precision loss in the geodesic distance computation through edge flipping. Additionally, when flipping geodesic paths using intrinsic triangles, the geodesic path initialization is required. Our current implementation initializes paths independently for each vertex pair during geodesic computation, neglecting potential shared information between different paths that could be exploited to improve performance. Finally, the method we use to obtain geodesic paths based on intrinsic triangle flipping relies on initialization and may not always yield the shortest geodesic length in highly complex scenarios.

For future work, we plan to enhance our algorithm in two primary directions. First, we aim to develop more appropriate strategies for geodesic length computation during intrinsic triangle flipping operations on parametric surfaces, thereby improving the accuracy of our geodesic computations. Second, we intend to eliminate redundant calculations by leveraging shared information between paths, further accelerating the overall computation process.

\section*{Acknowledgments}
The authors would like to thank the anonymous reviewers for their valuable comments and suggestions. This work was supported by the National Key R\&D Program of China (2022YFB3303200), and the National Natural Science Foundation of China (U23A20312, 62272277).

\section*{References}

\bibliography{mybibfile}

\end{document}